\documentclass[onecolumn,11pt]{IEEEtran}
\usepackage{palatino}

\usepackage{epsfig,subcaption, floatrow, calc}
\usepackage{amssymb, amsmath, amsthm}
\usepackage{color}

\usepackage{cite}
\usepackage{graphicx}
\DeclareGraphicsExtensions{.pdf,.jpeg,.png}
\usepackage{tikz, tikz-3dplot, xcolor}
\usepackage{marvosym}
\usepackage{tikzsymbols}
\usetikzlibrary{patterns}
\usetikzlibrary{matrix}
\usetikzlibrary{tikzmark}
\usetikzlibrary{fit}
\usetikzlibrary{shadows.blur}
\usetikzlibrary{shapes.symbols}
\usetikzlibrary{decorations.pathreplacing}
\usepackage{menukeys}

\usepackage{amssymb}
\usepackage{amsmath,amsfonts,amssymb}
\usepackage{verbatim}
\usepackage[bookmarks=false]{}

\usepackage{enumerate, enumitem}

\newtheorem{theorem}{Theorem}
\newtheorem{corollary}{Corollary}

\newtheorem{lemma}{Lemma}

\newtheorem{remark}{Remark}

\def \balpha {\boldsymbol{\alpha}}
\def \tbalpha {\tilde{\boldsymbol{\alpha}}}
\def \bbeta {\mathcal{S}}
\def \usr{\mathcal{S}}
\def \bbr {\boldsymbol{r}}
\def \tbbr {\tilde{\boldsymbol{r}}}
\def \dsumo {D_{\Sigma, o}}
\def \dsumb {D_{\Sigma, b}}

\def \Rsumo {R_{\Sigma, o}}
\def \Rsumb {R_{\Sigma, b}}

\usepackage[cmintegrals]{newtxmath}
\allowdisplaybreaks

\begin{document}

\title{The Extremal GDoF Gain of Optimal versus Binary Power Control in $K$ User Interference Networks Is $\Theta(\sqrt{K})$}

\author{Yao-Chia Chan, Pouya Pezeshkpour, Chunhua Geng, and Syed A. Jafar
\thanks{Y.-C. Chan and S. A. Jafar are with Center for Pervasive Communications and Computing (CPCC), University of California Irvine, Irvine, CA 92697 USA. Email: yaochic@uci.edu; syed@uci.edu. 
	P. Pezeshkpour is with University of California Irvine, Irvine, CA 92697 USA. Email: pezeshkp@uci.edu. 
	C. Geng is with MediaTek USA Inc., Irvine, CA 92618 USA. Email: chunhua.geng@mediatek.com. 
	This works is in part supported by funding from NSF grant CCF-1907053, ONR grant N00014-21-1-2386, and ARO grant W911NF1910344.}%
}

\maketitle

\begin{abstract}
	Using  ideas from Generalized Degrees of Freedom (GDoF) analyses and extremal network theory, this work studies the extremal gain of  optimal power control over binary (on/off) power control,  especially in large interference networks, in search of new theoretical insights. Whereas numerical studies have already established that in most practical settings binary power control is close to optimal, the extremal analysis shows not only that there exist settings where the gain from optimal power control can be quite significant, but also  bounds the extremal values of such gains from a GDoF perspective. As its main contribution, this work explicitly characterizes the extremal GDoF gain of optimal over binary power control as $\Theta\left(\sqrt{K}\right)$ for all $K$. In particular, the extremal gain is bounded between $\lfloor \sqrt{K}\rfloor$ and $2.5\sqrt{K}$ for every $K$. For $K=2,3,4,5,6$ users, the precise extremal gain is found to be $1, 3/2, 2, 9/4$ and $41/16$, respectively. Networks shown to achieve the extremal gain may be interpreted as multi-tier heterogeneous networks. It is worthwhile to note that because of their focus on asymptotic analysis, the sharp characterizations of extremal gains are valuable primarily from a theoretical perspective, and not as contradictions to the conventional wisdom that binary power control is generally close to optimal in  practical, non-asymptotic settings.
	
		
\end{abstract}


%
\IEEEpeerreviewmaketitle

\section{Introduction}
The design of communication systems involves many optimization problems. The inherent combinatorial complexity of these optimizations often makes direct analysis infeasible, and motivates complementary approaches from  theoretical and practical perspectives. On one hand, useful practical insights may be obtained from extensive numerical simulations. On the other hand, theoretical understanding may be advanced through extremal analyses, focused on identifying key cornerstones, e.g., asymptotic limits that are theoretically insightful. For instance,  information theory  often relies on extremal analysis (e.g., bounding the performance of the best code among all codes) under asymptotic (unbounded code length, unbounded complexity) relaxations to obtain  theoretically insightful answers (channel capacity) that are  complemented by numerical simulations to yield efficient communication designs. 

As the focus shifts from channels to networks, the combinatorial challenge is further compounded. Additional tools have  emerged to facilitate extremal/asymptotic theoretical analysis of wireless networks. Among these tools are the high SNR perspective taken by Generalized Degrees of Freedom (GDoF) studies \cite{Etkin_Onebit}, and more recently the idea of extremal network theory introduced in \cite{Chan_extremal}, i.e., the study of the extremal network settings that maximize the gain of one coding scheme over another. Extremal network theory and  practical simulations complement each other --- the latter focuses on what is likely to happen in settings of current practical interest, while the former tries to understand what is the best/worst scenario possible. Extremal network theory is particularly valuable for large networks that are  beyond the reach of numerical simulations.  For example, \cite{Chan_extremal} studies the extremal GDoF gain from joint coding (transmitter cooperation) over separate coding in a $K$ user interference network with finite precision channel state information at the transmitters (finite precision CSIT), and shows that in the parameter regimes known as TIN, CTIN and SLS, the extremal gain is $1.5$, $2-1/K$, and $\Theta(\log(K))$, respectively. Remarkably, \cite{Chan_STIN} finds that if secrecy constraints  are imposed, then there is no gain from transmitter cooperation in the SLS regime, but the extremal gain grows unbounded when the regime is relaxed to the STIN regime. Beyond the $K$ user interference channel, extremal GDoF analysis has  been applied by Joudeh and Caire in \cite{Joudeh_Caire_extremal} to downlink cellular networks, modeled as interfering broadcast channels, to characterize the gains of cooperation among base stations. Building on these advances, in this work we apply the GDoF and extremal network theory framework to an aspect that is essential to wireless network design -- power control \cite{Chiang_PC, Zander_PC,Zander_DPC, Foschini_DPC, Yates_SIF, Bambos_DPC}. In particular, we explore the extremal GDoF gain of optimal power control versus binary (on/off) power control in a (large) $K$ user interference network, where all interference is treated as (Gaussian) noise. 

Power control in conjunction with treating interference as noise (TIN) is one of the most widely prevalent ideas in wireless networks with multiple interfering links, because of its relative simplicity and robustness  compared with other sophisticated schemes, such as rate-splitting \cite{Han_Kobayashi, Clerckx_ratesplitting}, structured codes \cite{Jafar_symmetricIC, Bresler_lattice}, and interference alignment \cite{Jafar_IA}. It has been widely adopted for interference management in wireless local area networks \cite{wifi1} and homo-/heterogeneous cellular networks \cite{survey_muticell}, and remains one of the main tools for interference management in  emerging applications such as device-to-device (D2D) communications \cite{survey_D2D}, vehicle-to-everything (V2X) \cite{survey_V2X} networks, and space-air-ground integrated networks \cite{survey_UAV}. With power control and TIN, the rate $R_k$ achieved by User $k$, where $k\in \{1,2,\cdots, K\}$, is simply expressed as a function of the user's signal to interference and noise power ratio SINR$_k$,  as $R_k=\log(1+\text{SINR}_k)$, which gives rise to the problem of rate-optimization through power control. 
This problem has been the subject of extensive research. In general, the rate and power optimization problem is non-convex and NP-hard \cite{Luo_Spectrum,MAPEL}. 
Lower complexity suboptimal algorithms have been developed based on 
fixed-point iteration \cite{Chiang_Sumrate}, iterative water-filling \cite{Wei_waterfilling, Evans_waterfilling}, game theory \cite{Poor_GT}, geometric programming \cite{Chiang_GP}, weighted mean square error \cite{Luo_WMMSE}, fractional programming \cite{Wei_fraction}, and machine-learning approaches \cite{Sidiropoulos_pwrDNN, Wei_pwrDNN,  Cho_pwrCNN, Ribeiro_pwrGNN, Letaief_pwrGNN}. On the other hand, closed form characterizations of maximum sum-rate achievable with optimal power control remain rare, with analytical efforts mostly limited to  small networks \cite{2user_1,2user_2}, or symmetric settings  \cite{Hanly_Symmetric}. 

As an alternative to optimal power control, binary power control is also of much interest due to its lower implementation complexity. Binary power control refers to the transmission scheme where each transmitter operates at one of  two power levels: it is either switched off completely, or it transmits at full power. Optimal binary power control is essentially equivalent to a scheduling scheme, with the only choice corresponding to the selection of active users. The  rate region achievable with binary power control (and time-sharing) is characterized in \cite{Chara_TIN}. Rate optimization within this rate region is essentially an NP-hard non-convex integer programming problem \cite{S-MAPEL}. 
Lower complexity scheduling algorithms are developed based on message-passing \cite{Lai_BPCmsgpassing},  SINR heuristics \cite{FlashLinQ, Chenwei_BPC},   information-theoretic insights \cite{Naderi_Avestimehr_ITLinQ, Yi_Caire_ITLinQ+}, fractional programming \cite{Wei_FPLinQ}, and machine learning approaches \cite{Geoffery_schedulingGNN, Yu_schedulingDNN, Santiago_schedulingGNN}.
Binary power control is  found to be optimal for sum-rate maximization in many cases, e.g., multiple access channels\cite{Hanly_BPC_uplink}, $2$-user interference channels with single carrier \cite{2user_1,2user_2, BPC_FullDuplex} and multiple carriers \cite{BPC_Multicarrier}, $K$-user one-sided symmetric Wyner-type interference channels \cite{Hanly_BPC}, networks where the transmission rate is an artificial linear function of the received power \cite{linear_BPC}, and networks where either a geometric-arithmetic mean or low-SINR approximation is applicable \cite{2user_2}. Remarkably, numerical simulations of common communication network topologies such as cellular networks and D2D networks \cite{Alouini_cell, Hong_D2D, Andrews_D2D, Schober_D2D, Lozano_D2D} suggest that binary power control  has performance comparable to optimal power control. However, there are no known theoretical bounds on the worst-case sub-optimality penalty of binary power control relative to optimal power control. The extremal analysis undertaken in this paper is aimed at finding such bounds.

Some of the most analytically tractable characterizations of approximate rate regions with power control and TIN have emerged out of  high-SNR analysis in the GDoF framework \cite{Etkin_Onebit}.  
The GDoF region of the $K$ user interference channel with power control and TIN is known to be a union of polyhedra determined by the channel strength parameters \cite{Geng_TIN}. 
The region is in general not convex unless  time-sharing is allowed \cite{Yi_TIN}. 
Algorithms for joint rate assignment and power control under the GDoF framework are developed in \cite{Geng_PC,Geng_duality,Geng_TIM}, and GDoF analyses have also inspired practical spectrum sharing solutions such as ITLinQ \cite{Naderi_Avestimehr_ITLinQ} and ITLinQ+ \cite{Yi_Caire_ITLinQ+}. Notably, under certain parameter regimes, power control and TIN are known to be information theoretically optimal in the GDoF sense \cite{Geng_TIN,Yi_TIN,Chan_extremal}. Nevertheless, despite their relative simplicity, even the GDoF characterizations suffer from the curse of dimensionality as we study large networks. For a $K$ user interference channel, the GDoF region achievable by power control and TIN is a union of an exponential (in $K$) number of regions (corresponding to all subsets of active users), each characterized by an exponential number of so-called cycle bounds, which makes it challenging to extract fundamental insights about large networks from their GDoF regions. This is the challenge that motivates the extremal network theory perspective introduced in \cite{Chan_extremal}, which we use in this work to compare optimal power control with binary power control, especially for large networks.  Specifically  we  explore the extremal gain, defined by the supremum of the ratio of the sum GDoF corresponding to optimal and binary power control, across all possible network topologies. Our main result is that the extremal gain of optimal over binary power control is $\Theta(\sqrt{K})$. We also explicitly identify a class of network topologies that allow this extremal gain asymptotically at high SNR. Remarkably, these topologies  may be interpreted as  multi-tier heterogeneous networks. Thus, the extremal analysis adds surprising theoretical insights to the picture presented by numerical simulations --- whereas the numerical studies shows that binary power control is close to optimal for most commonly studied network topologies,  the theoretical analysis reveals exceptional network topologies where binary power control can be significantly suboptimal, and presents a tight (orderwise) bound on this suboptimality penalty. Interestingly, for smaller interference networks, with $K=2,3,4,5, 6$ users, we explicitly characterize the exact extremal GDoF gain of optimal over binary power control as $1, 3/2, 2, 9/4, 41/65$, respectively.

\emph{Notation}: For a real number $b$, $(b)^+$ and $\{b\}^+$ denote $\max\{0,b\}$. 
For a set $A$, $|A|$ denotes its cardinality. The set $\mathbb{N}$ collects all positive integers. For a positive integer $N$, define $[N] = \{1,2,\cdots, N\}$.
For two functions $f(x), g(x)$, $f(x) = \Theta(g(x))$ if $c_1 \leq \liminf_{x \rightarrow \infty} \left| \frac{f(x)}{g(x)} \right| \leq \limsup_{x \rightarrow \infty} \left| \frac{f(x)}{g(x)} \right| \leq c_2$ for some constant $c_1, c_2 >0$.

\section{System Model} \label{sec:model}

We consider a $K$-user Gaussian interference channel where each of the transmitters and receivers is equipped with a single antenna.
The signal observed by Receiver $k$ in the $t$-th channel use is,
\begin{align} \label{eq:GaussianModel}
	Y_k(t) &= \sum_{i=1}^{K} |h_{ki}| e^{j\theta_{ki}} \tilde{X}_i(t) + Z_k(t), &&\forall t \in [T], k \in [K],
\end{align}
where $h_{ki} \triangleq |h_{ki}| e^{j\theta_{ki}}$ is the channel coefficient between Transmitter $i$ and Receiver $j$, $\tilde{X}_i(t)$ is the complex-valued symbol sent by Transmitter $i$, subject to power constraint $\mathbb{E}[ |\tilde{X}_i(t)|^2] \leq P_i$, and $Z_k(t) \sim \mathcal{CN}(0,1)$ is the additive white Gaussian noise (AWGN) at Receiver $k$. 
We assume  perfect channel state information at the receivers (CSIR), while the transmitters know at least the magnitudes of the channel coefficients.

Following the standard GDoF formulation \cite{Etkin_Onebit, Geng_TIN}, we translate (\ref{eq:GaussianModel}) into the following,
\begin{align}\label{eq:model}
	Y_k(t) &= \sum_{i=1}^{K} \sqrt{P^{\alpha_{ki}}} X_i(t) + Z_k(t), &&\forall t \in [T], k \in [K].
\end{align}
Here $X_i(t)$ is the complex-valued symbol sent from Transmitter $i$, subject to unit power constraint, i.e., $ \mathbb{E}[ |{X}_i(t)|^2] \leq 1$. The parameter $\alpha_{ki} \triangleq (\log |h_{ki}|^2 P_i)^+ \geq 0$ 
is the channel strength parameter for the link between Transmitter $i$ and Receiver $k$.
Note that the phase terms $ e^{j\theta_{ki}}$ are omitted in (\ref{eq:model}), as they are inconsequential to the scheme of power control and treating interference as noise \cite{Geng_TIN} which is the focus of this work. For compact notation, let us collect all $\alpha_{ij}$ into $\balpha$ defined as
\begin{align} \label{eq:alpha}
	\balpha &=[\alpha_{ij}]_{i\in[K],j\in[K]},
\end{align}
which we will refer to as the network topology.

Next we formalize the scheme of power control and TIN.
Message $W_i\in\mathcal{W}_i$ originating at Transmitter $i$ is encoded  into codeword $X_i(t)$, $t\in[T]$,   with Gaussian codebooks of rate $R_i = \frac{1}{T} \log_2 |\mathcal{W}_i|$ and power  $\mathbb{E}[|X_{t}|^2] = P^{r_i}\leq 1$, where $r_i \leq 0$ are called the power allocation variables. Let us define the vector of power allocation variables as
\begin{align}
	\bbr&\triangleq (r_1, r_2, \cdots, r_K).
\end{align}

At each Receiver $i$, $i\in[K]$, the desired message is decoded by treating interference as noise, so that the following rate is achieved.
\begin{align}
	R_i(\balpha, P,\bbr) = \log_2 \left( 1 + \frac{ P^{\alpha_{ii}+r_i}}{ 1+ \sum_{k \in [K]\setminus\{i\}} P^{\alpha_{ik}+r_k} }  \right). \label{eq:tinRate}
\end{align}
We define $d_i(\balpha,\bbr)$ as the GDoF limit for the $i^{th}$ user, i.e.,
\begin{align}
	d_i(\balpha,\bbr) \triangleq \lim_{P \rightarrow \infty} \frac{R_i(\balpha,P,\bbr)}{\log_2 P}
	&= \left(\alpha_{ii} + r_i - \max_{k \in [K]\setminus\{i\}} \{ \alpha_{ik} + r_k \}^+ \right)^+, \label{eq:tinGDoF}
\end{align}
where the second equality comes from \cite[Eq. (7)]{Geng_TIN}.
The sum rate and sum GDoF achieved by power control and treating interference as noise are defined accordingly as,
\begin{align}
	{R}_\Sigma (\balpha, P, \bbr) &= \sum_{k=1}^{K} R_i(\balpha,P,\bbr),\\
	{D}_\Sigma (\balpha, \bbr) &= \sum_{k=1}^{K} d_i(\balpha,\bbr).
\end{align}

\subsection{Optimal and Binary Power Control} \label{sec:pwrctrl}
The main goal of this work is to compare Optimal Power Control (OPC) against Binary Power Control (BPC). Let us define these two power control schemes under the GDoF framework.
With OPC, the power exponents $r_k$ are allowed to take an arbitrary non-positive value, i.e., $r_k \leq 0$, for all $k \in [K]$.
On the other hand, Binary Power Control (BPC) allows each $r_i$ to take only binary values, $-\infty$ and $0$. These two values correspond to the two operating states of a transmitter: the former is equivalent to switching it off, while the latter amounts to transmitting at  full power. 	
The sum rate with OPC and BPC are respectively defined as 
\begin{align}
	\Rsumo (\balpha, P) &= \max_{\bbr \in \Omega^K} R_\Sigma(\balpha, P, \bbr), \label{eq:Rsumo}\\
	\Rsumb (\balpha, P) &= \max_{\bbr \in \Omega_b^K} R_\Sigma(\balpha, P, \bbr), \label{eq:Rsumb}
\end{align}
where $\Omega^K = \{ \bbr | -\infty \leq r_k \leq 0, \forall k \in [K]\}$, and $\Omega_b^K = \{ \bbr | r_k  \in \{ -\infty, 0 \}, \forall k \in [K]\}$.
The sum GDoF with OPC, $\dsumo$, and the one with BPC, $\dsumb$, are defined accordingly, i.e.,
\begin{align}
	\dsumo (\balpha) &= \max_{\bbr \in \Omega^K} D_\Sigma(\balpha, \bbr), \label{eq:dsumo}\\
	\dsumb (\balpha) &= \max_{\bbr \in \Omega_b^K} D_\Sigma(\balpha, \bbr).\label{eq:dsumb}
\end{align}

\section{Results: Extremal GDoF Gain} \label{sec:results}

The sum GDoF with OPC (\ref{eq:dsumo}) and the one with BPC (\ref{eq:dsumb}) are both functions of the network topology $\balpha$. The sum-GDoF gain, defined by the ratio of the sum GDoF achieved with OPC and BPC, also varies with different topologies. By taking the supremum over all possible topologies, we define\footnote{Note that $\mu_K$ is well-defined because it has an upper bound $K$, which is implied by the fact $\dsumo(\balpha) \leq K \max_{i \in [K]} \alpha_{ii}$ and the fact $\dsumb(\balpha) \geq \max_{i \in [K]} \alpha_{ii}$ for all $\balpha \in \mathcal{A}_K$.} the extremal GDoF gain for the $K$ user interference networks as,
\begin{align}
	\mu_K &= \sup_{\balpha \in \mathcal{A}_K } \frac{\dsumo(\balpha)}{\dsumb(\balpha)}, \label{eq:muK}
\end{align}
where $\mathcal{A}_K = \{ \balpha:~ \alpha_{ij} \geq 0, \forall i,j \in [K]\}$ denotes the set of all possible network topologies $\balpha$ for a $K$-user interference network.
In the following theorem we identify the exact value of the extremal GDoF gain for small networks with $6$ or fewer users.

\begin{theorem}\label{thm:2-6usr}
	For interference networks with $K=2,3,4,5,6$ users, the exact value of the extremal GDoF gain $\mu_K$ of optimal over binary power control  is listed in the following table.
	
	\begin{center}
		\begin{tabular}{ c | c | c | c | c | c }
			$K$ & $2$ & $3$ & $4$ & $5$ & $6$ \\ \hline
			$\mu_K$ & $1$ & $3/2$ & $2$ & $9/4$ & $41/16$ \\
		\end{tabular}
	\end{center}

\end{theorem}
The proof of Theorem \ref{thm:2-6usr} appears in Section \ref{sec:proof2-6usr}. 

\begin{enumerate}[wide, labelindent=1em ,labelwidth=!, labelsep*=1em, leftmargin =0em, style = sameline , label=\it Observation(\it\arabic*)]

	\item Network topologies that achieve the extremal GDoF gains $\mu_3=3/2, \mu_4=2, \mu_5=9/4$ and $\mu_6=41/16$ for $K=3,4,5,6$ users, respectively, are illustrated in Figure \ref{fig:extremalnet}. These are extremal networks, because according to Theorem \ref{thm:2-6usr} for each value of $K$, no other topology can achieve a higher GDoF gain from optimal power control over binary power control.
	\item While conventional wisdom backed by extensive numerical simulations has already established that for most networks binary power control tends to be close to optimal, evidently even for relatively smaller number of users there exist networks (as shown in Figure \ref{fig:extremalnet}) where optimal power control achieves significant (e.g., factor of $2$ for $K=4$ users) gains in throughput over binary power control.
	\item \label{obs:noguess} Apparently, $\mu_K$ is increasing with $K$. However, the rate of increase does not behave in a monotonic fashion. As $K$ takes values $2,3,4,5,6$, the successive increases in $\mu_K$ are $50\%, 33\%, 12.5\%, 13.89\%$, respectively. This observation shows the difficulty of guessing the exact general functional form of $\mu_K$ from the study of smaller networks.
	
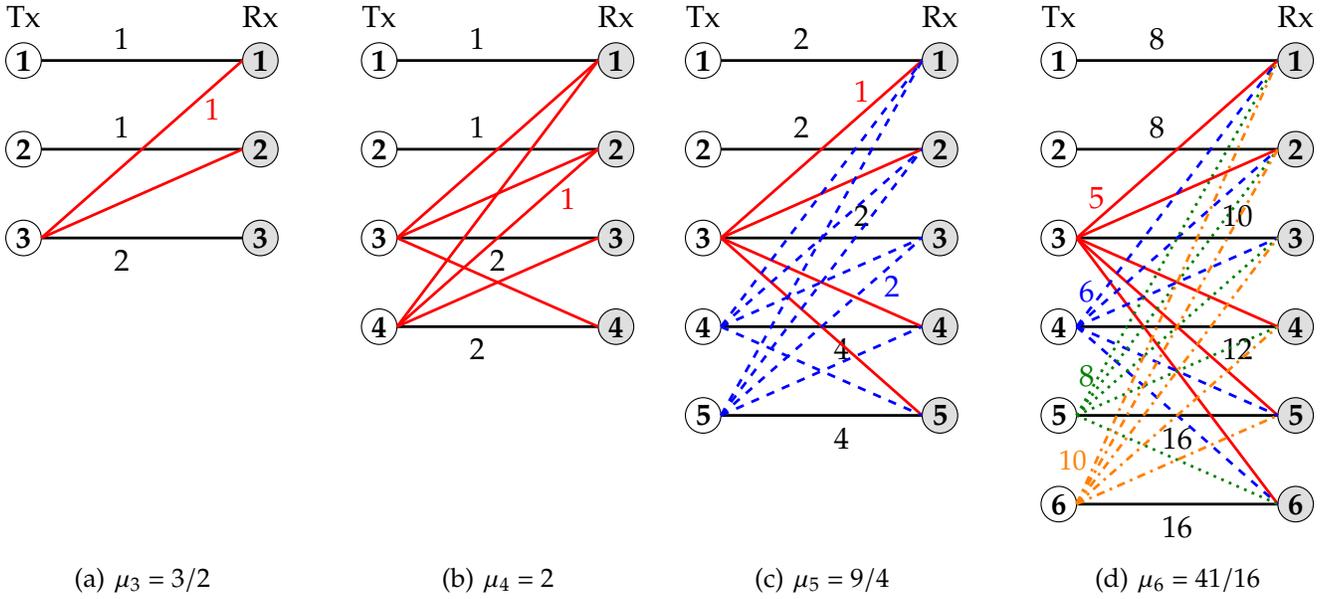
\begin{figure}[t!]
		
		\floatsetup{heightadjust = all, valign=t}
		\ffigbox[][]{%
			\begin{subfloatrow}
				\ffigbox[\FBwidth][8cm]{
					\resizebox{0.23\textwidth}{!}{
						\begin{tikzpicture}[scale=1.5]
							\def \r {0.15}
							\def \w {2}
							\def \d {0.75}
							
							\foreach \a in {1,2,3}{
								\node (T\a) at (0, {-1 * \a * \d}) {};
								\node (R\a) at (\w, {-1 * \a * \d}) {};
							}
							\node (Tx) at ($(T1) + (0, 0.5*\d)$) {Tx};
							\node (Rx) at ($(R1) + (0, 0.5*\d)$) {Rx};
							
							\foreach \a in {1,2,3}{
								\draw (T\a) circle (\r) node {$\boldsymbol{\a}$};
								\draw [fill = gray!25] (R\a) circle (\r) node {$\boldsymbol{\a}$};
							}
							
							\draw [line width = 1, black]($(T1)+(\r,0)$)--($(R1)-(\r,0)$) node [pos=0.4, above] {$1$};
							\draw [line width = 1, black]($(T2)+(\r,0)$)--($(R2)-(\r,0)$) node [pos=0.4, above] {$1$};
							\draw [line width = 1, black]($(T3)+(\r,0)$)--($(R3)-(\r,0)$) node [pos=0.4, below] {$2$};
							\draw [line width = 1, red]($(T3)+(\r,0)$)--($(R1)-(\r,0)$) node [pos=0.85, below, red] {$1$};
							\draw [line width = 1, red]($(T3)+(\r,0)$)--($(R2)-(\r,0)$);
						\end{tikzpicture}
				}}{\caption{$\mu_3=3/2$}}
				\ffigbox[\FBwidth][8cm]{
					\resizebox{0.23\textwidth}{!}{
						\begin{tikzpicture} [scale=1.5]
							
							\def \r {0.15}
							\def \w {2}
							\def \d {0.75}
							
							\foreach \a in {1,2,3,4}{
								\node (T\a) at (0, {-1 * \a * \d}) {};
								\node (R\a) at (\w, {-1 * \a * \d}) {};
							}
							\node (Tx) at ($(T1) + (0, 0.5*\d)$) {Tx};
							\node (Rx) at ($(R1) + (0, 0.5*\d)$) {Rx};
							
							\foreach \a in {1,2,3,4}{
								\draw (T\a) circle (\r) node {$\boldsymbol{\a}$};
								\draw [fill = gray!25] (R\a) circle (\r) node {$\boldsymbol{\a}$};
							}
							
							\draw [line width = 1, black]($(T1)+(\r,0)$)--($(R1)-(\r,0)$) node [pos=0.4, above] {$1$};
							\draw [line width = 1, black]($(T2)+(\r,0)$)--($(R2)-(\r,0)$) node [pos=0.4, above] {$1$};
							\draw [line width = 1, black]($(T3)+(\r,0)$)--($(R3)-(\r,0)$) node [pos=0.5	, below] {$2$};
							\draw [line width = 1, black]($(T4)+(\r,0)$)--($(R4)-(\r,0)$) node [pos=0.4, below] {$2$};
							\draw [line width = 1, red]($(T3)+(\r,0)$)--($(R1)-(\r,0)$);
							\draw [line width = 1, red]($(T3)+(\r,0)$)--($(R2)-(\r,0)$);
							\draw [line width = 1, red]($(T3)+(\r,0)$)--($(R4)-(\r,0)$);
							\draw [line width = 1, red]($(T4)+(\r,0)$)--($(R1)-(\r,0)$);
							\draw [line width = 1, red]($(T4)+(\r,0)$)--($(R2)-(\r,0)$) node [pos=0.85, below, red] {$1$};
							\draw [line width = 1, red]($(T4)+(\r,0)$)--($(R3)-(\r,0)$);
						\end{tikzpicture}
				}}{\caption{$\mu_4=2$}}
				\ffigbox[\FBwidth][8cm]{
					\resizebox{0.23\textwidth}{!}{
						\begin{tikzpicture} [scale=1.5]
							\def \r {0.15}
							\def \w {2}
							\def \d {0.75}
							
							\foreach \a in {1,2,3,4,5}{
								\node (T\a) at (0, {-1 * \a * \d}) {};
								\node (R\a) at (\w, {-1 * \a * \d}) {};
							}
							\node (Tx) at ($(T1) + (0, 0.5*\d)$) {Tx};
							\node (Rx) at ($(R1) + (0, 0.5*\d)$) {Rx};
							
							\foreach \a in {1,2,3,4,5}{
								\draw (T\a) circle (\r) node {$\boldsymbol{\a}$};
								\draw [fill = gray!25] (R\a) circle (\r) node {$\boldsymbol{\a}$};
							}
							
							\draw [line width = 1, black]($(T1)+(\r,0)$)--($(R1)-(\r,0)$) node [pos=0.4, above] {$2$};
							\draw [line width = 1, black]($(T2)+(\r,0)$)--($(R2)-(\r,0)$) node [pos=0.4, above] {$2$};
							\draw [line width = 1, black]($(T3)+(\r,0)$)--($(R3)-(\r,0)$) node [pos=0.7	, above] {$2$};
							\draw [line width = 1, black]($(T4)+(\r,0)$)--($(R4)-(\r,0)$) node [pos=0.6, below] {$4$};
							\draw [line width = 1, black]($(T5)+(\r,0)$)--($(R5)-(\r,0)$) node [pos=0.6, below] {$4$};
							
							\draw [line width = 1, red]($(T3)+(\r,0)$)--($(R1)-(\r,0)$) node [pos=0.7, above] {$1$};
							\foreach \a in {2,4,5}{
								\draw [line width = 1, red]($(T3)+(\r,0)$)--($(R\a)-(\r,0)$);
							}
							
							\foreach \a in {1,2,3,5}{
								\draw [line width = 1, blue, dashed]($(T4)+(\r,0)$)--($(R\a)-(\r,0)$);
							}
							\foreach \a in {1,2,4}{
								\draw [line width = 1, blue, dashed]($(T5)+(\r,0)$)--($(R\a)-(\r,0)$);
							}
							\draw [line width = 1, blue, dashed]($(T5)+(\r,0)$)--($(R3)-(\r,0)$) node [pos=0.85, below] {$2$};
						\end{tikzpicture}
				}}{\caption{$\mu_5=9/4$}}
				\ffigbox[\FBwidth][8cm]{
					\resizebox{0.23\textwidth}{!}{
						\begin{tikzpicture} [scale=1.5]
							\def \r {0.15}
							\def \w {2}
							\def \d {0.75}
							
							\foreach \a in {1,2,3,4,5,6}{
								\node (T\a) at (0, {-1 * \a * \d}) {};
								\node (R\a) at (\w, {-1 * \a * \d}) {};
							}
							\node (Tx) at ($(T1) + (0, 0.5*\d)$) {Tx};
							\node (Rx) at ($(R1) + (0, 0.5*\d)$) {Rx};
							\foreach \a in {1,2,3,4,5,6}{
								\draw (T\a) circle (\r) node {$\boldsymbol{\a}$};
								\draw [fill = gray!25] (R\a) circle (\r) node {$\boldsymbol{\a}$};
							}
							
							\draw [line width = 1, black]($(T1)+(\r,0)$)--($(R1)-(\r,0)$) node [pos=0.4, above] {$8$};
							\draw [line width = 1, black]($(T2)+(\r,0)$)--($(R2)-(\r,0)$) node [pos=0.4, above] {$8$};
							\draw [line width = 1, black]($(T3)+(\r,0)$)--($(R3)-(\r,0)$) node [pos=0.8, above] {$10$};
							\draw [line width = 1, black]($(T4)+(\r,0)$)--($(R4)-(\r,0)$) node [pos=0.8, below] {$12$};
							\draw [line width = 1, black]($(T5)+(\r,0)$)--($(R5)-(\r,0)$) node [pos=0.5, below] {$16$};
							\draw [line width = 1, black]($(T6)+(\r,0)$)--($(R6)-(\r,0)$) node [pos=0.5, below] {$16$};
							
							\draw [line width = 1, red]($(T3)+(\r,0)$)--($(R1)-(\r,0)$) node [pos= 0.1, above] {5};
							\foreach \a in {2,4,5,6}{
								\draw [line width = 1, red]($(T3)+(\r,0)$)--($(R\a)-(\r,0)$);
							}
							
							\draw [line width = 1, blue, dashed]($(T4)+(\r,0)$)--($(R1)-(\r,0)$) node [pos= 0.05, above] {6};
							\foreach \a in {2,3,5,6}{
								\draw [line width = 1, blue, dashed]($(T4)+(\r,0)$)--($(R\a)-(\r,0)$);
							}
							
							\draw [line width = 1, green!50!black, dotted]($(T5)+(\r,0)$)--($(R1)-(\r,0)$)  node [pos= 0.05, above] {8};
							\foreach \a in {2,3,4,6}{
								\draw [line width = 1, green!50!black, dotted]($(T5)+(\r,0)$)--($(R\a)-(\r,0)$);
							}
							
							\draw [line width = 1, orange, dash dot]($(T6)+(\r,0)$)--($(R1)-(\r,0)$)node [pos= 0.1, above, left = -0.5] {\small 10};
							\foreach \a in {2,3,4,5}{
								\draw [line width = 1, orange, dash dot]($(T6)+(\r,0)$)--($(R\a)-(\r,0)$);
							}
						\end{tikzpicture}		
				}}{\caption{$\mu_6=41/16$}}
			\end{subfloatrow}
		}{\caption{\small \it $K$ user interference networks achieving extremal GDoF gain $\mu_K=3/2, 2, 9/4, 41/16$, respectively, for (a) $K=3$, (b) $K=4$, (c) $K=5$, and (d) $K=6$ users.		}
			\label{fig:extremalnet}}
	\end{figure}

\end{enumerate}

The extremal sum-GDoF gain metric $\mu_K$   is grounded in asymptotic (high-SNR) analysis because of analytical tractability and the potential for sharp insights that are difficult to obtain directly in finite-SNR regimes, especially for large networks. Intuitively, because extremal gains are  likely to manifest in high SNR regimes in any case, we expect that the extremal sum-GDoF gain should closely reflect the unconstrained extremal sum-rate gain across \emph{all} SNR regimes, i.e., without the restriction of $P\rightarrow\infty$. The following corollary of  Theorem \ref{thm:2-6usr} confirms this intuition for smaller networks. The proof of Corollary \ref{cor:2-4usr} is relegated to Section \ref{sec:proofCor}.

\begin{corollary}\label{cor:2-4usr}
	Consider the extremal sum-rate gain, which is defined for $K$ user interference networks as follows,
	\begin{align}
		\eta_K = \sup_{\balpha \in \mathcal{A}_K, P \geq 0 } \frac{\Rsumo(\balpha,P)}{\Rsumb(\balpha,P)}. \label{eq:etaK}
	\end{align}
	The extremal sum-rate gain $\eta_K$ is equal to the extremal sum-GDoF gain $\mu_K$ for For $K=2,3,4$, i.e.,
	
	\begin{center}
		\begin{tabular}{ c | c | c | c}
			$K$ & $2$ &$3$ &$4$ \\ 
			\hline
			$\eta_K$ & $1$ &$\frac{3}{2}$ &$2$
		\end{tabular}
	\end{center}

\end{corollary}

Next we proceed to the main goal of this work, i.e., to understand the extremal GDoF gain for large interference networks. Since the general form of the $\mu_K$ function is not readily apparent from the studies of small networks, we must directly explore large interference networks. Our main result  is stated in the following theorem.

\begin{theorem}\label{thm:OsqrtK}
	For $K$-user interference channels, the extremal  sum-GDoF gain of optimal over binary power control, $\mu_K$, satisfies $\mu_K = \Theta(\sqrt{K})$. Specifically,
	\begin{align}
		1 \leq \liminf_{K \rightarrow \infty} \frac{\mu_K}{\sqrt{K}} \leq \limsup_{K \rightarrow \infty} \frac{\mu_K}{\sqrt{K}} \leq 5/2.
	\end{align}
\end{theorem}

The proof of Theorem \ref{thm:OsqrtK} is provided in Section \ref{sec:proofSqrtK}. 

\begin{enumerate}[wide, labelindent=1em ,labelwidth=!, labelsep*=1em, leftmargin =0em, style = sameline , label=\it Observation(\it\arabic*)]
	\setcounter{enumi}{3}
	\item The proof of Theorem \ref{thm:OsqrtK} reveals the bounds $\lfloor\sqrt{K}\rfloor \leq \mu_K \leq 2.5\sqrt{K}$ for all $K\in\mathbb{N}$. Notably, these bounds are not only valid for large $K$, rather they hold for all $K$. The bounds can be further tightened to  within a factor of $1.5$ of each other when $K$ is a perfect square.
	
	\item  An extremal network for $K=m^2$ users is shown in Figure \ref{fig:extremalK} and has a hierarchical structure, with $m$ groups, each of which is comprised of $m$ users. The strength ($\alpha_{ij}$ values) of the interfering links emanating from the $g^{th}$ group of users, $g\in[m]$, is equal to $g-1$, as experienced by users in groups $1,2,\cdots, g$, while users in groups $g+1,\cdots, m$ see no interference from users in group $g, g-1, \cdots, 1$. The desired links of users in group $g$  all have strength $\alpha_{ii}=g$.
	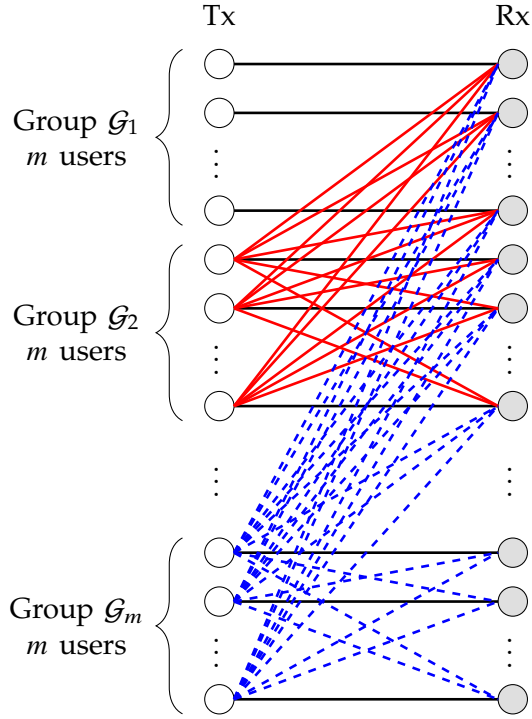
\begin{figure}[!t]
		\centering
		\begin{tikzpicture}[scale=1.3]
			\def \r {0.15}
			\def \w {3}
			\def \d {0.5}
			
			\node (T1) at (0, {-1 * 0 * \d}) {};
			\node (T2) at (0, {-1 * 1 * \d}) {};
			\node (Tdot1) at (0, {-1 * 2 * \d}) {};
			\node (T3) at (0, {-1 * 3 * \d}) {};
			\node (T4) at (0, {-1 * 4 * \d}) {};
			\node (T5) at (0, {-1 * 5 * \d}) {};
			\node (Tdot2) at (0, {-1 * 6 * \d}) {};
			\node (T6) at (0, {-1 * 7 * \d}) {};
			\node (Tdot3) at (0, {-1 * 8.5 * \d}) {};
			\node (T7) at (0, {-1 * 10 * \d}) {};
			\node (T8) at (0, {-1 * 11 * \d}) {};
			\node (Tdot4) at (0, {-1 * 12 * \d}) {};
			\node (T9) at (0, {-1 * 13 * \d}) {};
			
			\node (R1) at (\w, {-1 * 0 * \d}) {};
			\node (R2) at (\w, {-1 * 1 * \d}) {};
			\node (Rdot1) at (\w, {-1 * 2 * \d}) {};
			\node (R3) at (\w, {-1 * 3 * \d}) {};
			\node (R4) at (\w, {-1 * 4 * \d}) {};
			\node (R5) at (\w, {-1 * 5 * \d}) {};
			\node (Rdot2) at (\w, {-1 * 6 * \d}) {};
			\node (R6) at (\w, {-1 * 7 * \d}) {};
			\node (Rdot3) at (\w, {-1 * 8.5 * \d}) {};
			\node (R7) at (\w, {-1 * 10 * \d}) {};
			\node (R8) at (\w, {-1 * 11 * \d}) {};
			\node (Rdot4) at (\w, {-1 * 12 * \d}) {};
			\node (R9) at (\w, {-1 * 13 * \d}) {};

			\node (Tx) at ($(T1) + (0, \d)$) {Tx};
			\node (Rx) at ($(R1) + (0, \d)$) {Rx};

			\foreach \a in {1,2,3,4,5,6,7,8,9}{
				\draw (T\a) circle (\r);
				\draw [fill = gray!25] (R\a) circle (\r);
			}
			\foreach \a in {1,2,3,4}{
				\draw (Tdot\a) node [rotate = 90] {$\cdots$};
				\draw (Rdot\a) node [rotate = 90] {$\cdots$};
			}
			
			\draw [decorate,decoration={brace,amplitude=10pt, mirror},xshift=-15,yshift=0pt] ({\r},\r) -- ({\r},{-3*\d-\r}) node [black,midway,xshift=-40, align=center] {Group $\mathcal{G}_1$ \\ $m$ users};
			\draw [decorate,decoration={brace,amplitude=10pt, mirror},xshift=-15,yshift=0pt] ({\r},{-4*\d+\r}) -- ({\r},{-7*\d-\r}) node [black,midway,xshift=-40, align=center] {Group $\mathcal{G}_2$ \\ $m$ users};
			\draw [decorate,decoration={brace,amplitude=10pt, mirror},xshift=-15,yshift=0pt] ({\r},{-10*\d+\r}) -- ({\r},{-13*\d-\r}) node [black,midway,xshift=-40, align=center] {Group $\mathcal{G}_m$ \\ $m$ users};
			\foreach \a in {1,2,3,4,5,6,7,8,9}{
				\draw [line width = 1, black]($(T\a)+(\r,0)$)--($(R\a)-(\r,0)$);
			}
			\foreach \a in {1,2,3,5,6}{
				\draw [line width = 1, red]($(T4)+(\r,0)$)--($(R\a)-(\r,0)$);
			}
			\foreach \a in {1,2,3,4,6}{
				\draw [line width = 1, red]($(T5)+(\r,0)$)--($(R\a)-(\r,0)$);
			}
			\foreach \a in {1,2,3,4,5}{
				\draw [line width = 1, red]($(T6)+(\r,0)$)--($(R\a)-(\r,0)$);
			}
			\foreach \a in {1,2,3,4,5,6,8,9}{
				\draw [line width = 1, blue, dashed] ($(T7)+(\r,0)$)--($(R\a)-(\r,0)$);
			}
			\foreach \a in {1,2,3,4,5,6,7,9}{
				\draw [line width = 1, blue, dashed] ($(T8)+(\r,0)$)--($(R\a)-(\r,0)$);
			}
			\foreach \a in {1,2,3,4,5,6,7,8}{
				\draw [line width = 1, blue, dashed] ($(T9)+(\r,0)$)--($(R\a)-(\r,0)$);
			}
			
		\end{tikzpicture}
		\caption{\small \it An interference network with $K=m^2$ users, that achieves sum-GDoF gain of $m=\sqrt{K}$ from optimal over binary power control.}\label{fig:extremalK}
	\end{figure}
	
\end{enumerate}

\section{Key Lemmas}

Define a subset of all possible topologies $\mathcal{A}_K$ as,
\begin{align}
	\mathcal{A}_K^+ &\triangleq \left\{ \balpha:~  \left(\exists \bbr\in\Omega^K \text{ s.t. }  \dsumo (\balpha) = \sum_{i=1}^{K} d_i(\balpha,\bbr)\right) \Rightarrow \left(\min_{i\in[K]}d_{i}(\balpha,\bbr)>0\right) \right\},
\end{align}
which  collects the $K$ user network topologies whose  sum GDoF under OPC, $\dsumo(\balpha)$ can be achieved $\emph{only}$ by GDoF tuples with non-zero components. 
For example, the topology $\balpha = [\alpha_{11},\alpha_{12}; \alpha_{21},\alpha_{22}] = [1, 0.25; 0.25, 1]$ for $K=2$ users, lies in $\mathcal{A}_2^+$, because its sum-GDoF value with OPC, $\dsumo(\balpha) = 1.5$ can be achieved only with GDoF tuples $(d_1,d_2) = (1, 0.5)$ or $(0.5, 1)$, or any tuple lying on the line segment connecting these two tuples. All of these tuples have non-zero components.
On the other hand, the topology  $\balpha = [1, 0.5; 0.5, 1]$ does not lie in $\mathcal{A}_2^+$, because its sum GDoF, $\dsumo(\balpha) = 1$, can be achieved by the tuple $(d_1, d_2) = (1,0)$.

Before proceeding to the proof of Theorem \ref{thm:2-6usr} and Theorem \ref{thm:OsqrtK}, we first present some observations and a key lemma that we will need for the converse arguments.  As our first observation, the extremal gain found within $\mathcal{A}_K^+$ is non-decreasing in $K$.

\begin{lemma}\label{lemma:increase}
	For $K \in \mathbb{N}$,
	\begin{align}
		\sup_{\balpha \in \mathcal{A}_K^+} \frac{\dsumo(\balpha)}{\dsumb(\balpha)} \leq \sup_{\balpha' \in \mathcal{A}_{K+1}^+} \frac{\dsumo(\balpha')}{\dsumb(\balpha')}.\label{eq:plusplus}
	\end{align}
\end{lemma}
\begin{proof} The proof is straightforward because for any $K$ user topology $\balpha\in\mathcal{A}_K^+$, we can add a $(K+1)^{th}$ user who neither causes nor suffers interference from any of the original $K$ users $(\alpha_{i,K+1}=\alpha_{K+1,i}=0,\forall i\in[K])$, and has a desired channel strength $\alpha_{K+1,K+1}=\epsilon>0$, to obtain a $K+1$ user topology, $\balpha'\in\mathcal{A}_{K+1}^+$. Since this $(K+1)^{th}$ user must contribute $\epsilon>0$ to the sum-GDoF optimal solution with either OPC or BPC, $\dsumo(\balpha')=\dsumo(\balpha)+\epsilon$, and $\dsumb(\balpha')=\dsumb(\balpha)+\epsilon$. Since $\epsilon$ can be chosen to be  arbitrarily small, the RHS of \eqref{eq:plusplus} cannot be smaller than the LHS. 
\end{proof}

The next lemma bounds the extremal GDoF gain from above by two upper bounds in two complementary topology subsets.

\begin{lemma} \label{lemma:mubound}
	For $K\geq 2$,
	\begin{align}
		\mu_K \leq \max \left\{ \mu_{K-1}, \sup_{\balpha \in \mathcal{A}_K^+} \frac{\dsumo(\balpha)}{\dsumb(\balpha)}  \right\}. \label{eq:mubound}
	\end{align}
\end{lemma}
\begin{proof} If $\mu_K=\sup_{\balpha \in \mathcal{A}_K^+} \frac{\dsumo(\balpha)}{\dsumb(\balpha)} $ then \eqref{eq:mubound} holds. So let us assume the alternative, i.e., $$\mu_K=\sup_{\balpha \in \mathcal{A}_K\setminus \mathcal{A}_K^+} \frac{\dsumo(\balpha)}{\dsumb(\balpha)}.$$ Consider any $K$ user topology $\tilde\balpha_K \in \mathcal{A}_K\setminus \mathcal{A}_K^+$. By definition, the sum-GDoF value $\dsumo(\tilde\balpha_K)$ must be achieved with  $d_i(\tilde\balpha_K,\bbr)=0$ for some $i\in[K]$. Removing this user $i$, we obtain a $K-1$ user topology $\tilde\balpha_{K-1}$. Since User $i$ contributed nothing to the sum-GDoF under OPC, removing this user cannot hurt, i.e., $\dsumo(\tilde\balpha_{K-1})=\dsumo(\tilde\balpha_{K})$. On the other hand, removing this user cannot help increase the sum-GDoF with BPC, because switching off this user was already an option in the original $K$ user topology under BPC. Therefore, $\dsumb(\tilde\balpha_{K-1})\leq \dsumb(\tilde\balpha_{K})$, and we have,
\begin{align}
	\frac{\dsumo(\tilde{\balpha}_K)}{\dsumb(\tilde{\balpha}_K)} \leq \frac{\dsumo(\tilde{\balpha}_{K-1})}{\dsumb(\tilde{\balpha}_{K-1})} \leq \mu_{K-1}.
\end{align}
Since this holds for all topologies $\tilde\balpha_K \in \mathcal{A}_K\setminus \mathcal{A}_K^+$, we have the bound $\mu_K\leq \mu_{K-1}$, which concludes the proof. 
\end{proof}

\begin{lemma} \label{lemma:pwr}
	For any topology $\balpha$ and any power allocation vector $\bbr=(r_1,r_2,\cdots,r_K)\in\Omega^K$, with $r_{\max}\triangleq \max_{k\in[K]}r_k$, 
	\begin{align}
		d_k(\balpha,\bbr)&=d_k(\balpha,\bbr'),&&\forall k\in[K],\label{eq:gdofhold}
	\end{align}
	where the new power allocation vector $\bbr'=(r'_1,r'_2,\cdots, r'_K)\in\Omega^K$ and 
	\begin{align}
		r'_k&=r_k-r_{\max},&&\forall k\in[K].
	\end{align}
\end{lemma}

\begin{proof}
	The transmit power levels ($r_k$, $k\in[K]$) of all users are elevated (e.g., as illustrated in Figure \ref{fig:pwr}) by the same amount ($-r_{\max}$), until at least one user hits its maximum transmit power level $(0)$. The elevated power allocation variables $r_i'$ are still valid $(\bbr'\in\Omega^K)$ because their values are not more than $0$. This can be seen as $r_i' = r_i - \max_{j \in [K]} r_j  \leq r_i - r_i = 0$. With the new power allocation variables $r_i'$, every receiver sees the same upward shift (increase by $-r_{\max}$) in its desired signal as well as all the interfering signals. The difference of power levels between desired signals and interference is therefore unaffected. As a result, the GDoF achieved by TIN are unchanged, i.e., \eqref{eq:gdofhold} holds.
\end{proof}

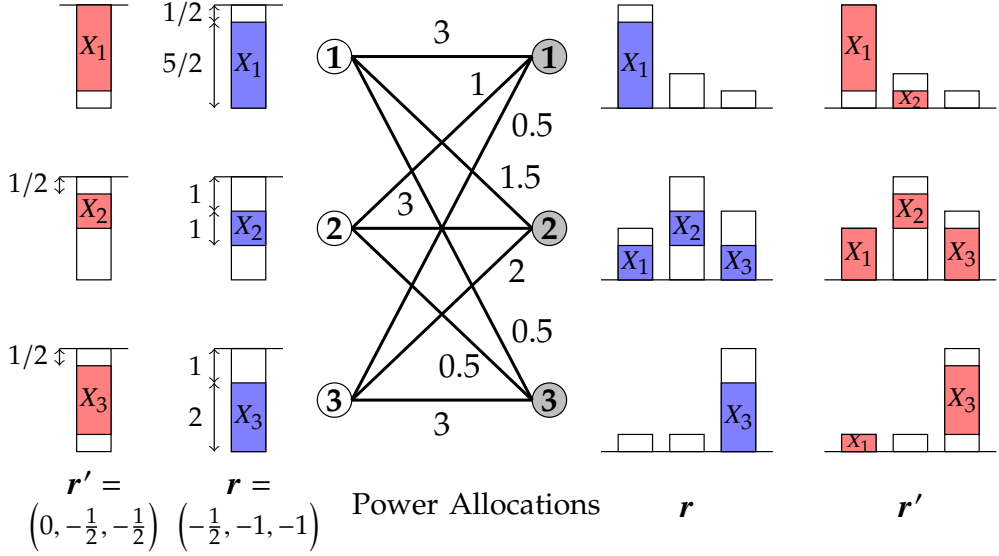
\begin{figure}[t!]
	\centering
	\resizebox{0.75\textwidth}{!}{
		\begin{tikzpicture}[scale=2]
			\def \r {0.1}
			\def \d {1}
			\def \w {1.25}
			\def \e {0.2}
			\def \h {0.6*\d}
			
			\foreach \v in {1,2,3}{
				\coordinate (Tx\v) at (0, {-1*\d*\v});
				\draw [fill=white] (Tx\v) circle (\r) node {$\boldsymbol{\v}$};
				\coordinate (Rx\v) at ({\w}, {-1*\d*\v});
				\draw [fill=gray!50] (Rx\v) circle (\r) node {$\boldsymbol{\v}$};
				\foreach \t in {1,2}{
					\coordinate (Tx\v-T\t) at ($(Tx\v) + ({\t*(-0.3) + (\t-1)*(-0.6)}, 0.3*\d)$);
					\coordinate (Rx\v-T\t) at ($(Rx\v) + ({\t*(0.3) + (\t-1)*(1.0)}, -0.3*\d)$);
				}
			}
			
			\draw [line width = 1] ($(Tx1) + (\r, 0)$) -- ($(Rx1) - (\r, 0)$) node [pos=0.5, above] {3};
			\draw [line width = 1] ($(Tx2) + (\r, 0)$) -- ($(Rx2) - (\r, 0)$) node [pos=0.3, above] {3};
			\draw [line width = 1] ($(Tx3) + (\r, 0)$) -- ($(Rx3) - (\r, 0)$) node [pos=0.5, below] {3};
			\draw [line width = 1] ($(Tx1) + (\r, 0)$) -- ($(Rx2) - (\r, 0)$) node [pos=0.7, above, right=3] {1.5};
			\draw [line width = 1] ($(Tx1) + (\r, 0)$) -- ($(Rx3) - (\r, 0)$) node [pos=0.8, above, right = 1] {0.5};
			\draw [line width = 1] ($(Tx2) + (\r, 0)$) -- ($(Rx1) - (\r, 0)$) node [pos=0.7, above] {1};
			\draw [line width = 1] ($(Tx2) + (\r, 0)$) -- ($(Rx3) - (\r, 0)$) node [pos=0.8, below, left = 1] {0.5};
			\draw [line width = 1] ($(Tx3) + (\r, 0)$) -- ($(Rx1) - (\r, 0)$) node [pos=0.8, below, right=1] {0.5};
			\draw [line width = 1] ($(Tx3) + (\r, 0)$) -- ($(Rx2) - (\r, 0)$) node [pos=0.75, below, right =3 ] {2};
			\node at ($  (Tx3)+({\w/2+0.2}, -0.6)$) {\small Power Allocations};

			\foreach \t in {1,2}{
				\foreach \v in {1,2,3}{
					\draw ($(Tx\v-T\t)$) -- ($(Tx\v-T\t) + (-0.4,0)$);
					\draw ($(Tx\v-T\t) + (-0.1,0)$) rectangle ($(Tx\v-T\t) + (-0.3, {-1*\h})$);
				}
			}
			\node at ($(Tx3-T1)+(-0.2, {-1*\h-0.65*\d})$) [align=center, anchor=south] { $\boldsymbol{r} = $\\[0ex] \footnotesize $\left (-\tfrac{1}{2}, -1, -1 \right)$};
			\node at ($(Tx3-T2)+(-0.2, {-1*\h-0.65*\d})$) [align=center, anchor=south] { $\boldsymbol{r}' = $\\[0ex] \footnotesize $\left(0, -\tfrac{1}{2}, -\tfrac{1}{2} \right)$};
			
			\draw [fill = blue!50]($(Tx1-T1)+(-0.1, {-1/6*\h})$) rectangle ($(Tx1-T1)+(-0.3, {-1*\h})$) node [pos=0.5] { \footnotesize $X_1$};
			\draw [fill = blue!50]($(Tx2-T1)+(-0.1, {-1/3*\h})$) rectangle ($(Tx2-T1)+(-0.3, {-2/3*\h})$) node [pos=0.5] {\footnotesize $X_2$};
			\draw [fill = blue!50]($(Tx3-T1)+(-0.1, {-1/3*\h})$) rectangle ($(Tx3-T1)+(-0.3, {-1*\h})$) node [pos=0.5] {\footnotesize $X_3$};
			\draw [fill = red!50]($(Tx1-T2)+(-0.1, {0*\h})$) rectangle ($(Tx1-T2)+(-0.3, {-5/6*\h})$) node [pos=0.5] { \footnotesize $X_1$};
			\draw [fill = red!50]($(Tx2-T2)+(-0.1, {-1/6*\h})$) rectangle ($(Tx2-T2)+(-0.3, {-1/2*\h})$) node [pos=0.5] {\footnotesize $X_2$};
			\draw [fill = red!50]($(Tx3-T2)+(-0.1, {-1/6*\h})$) rectangle ($(Tx3-T2)+(-0.3, {-5/6*\h})$) node [pos=0.5] { \footnotesize  $X_3$};

			\draw [<->] ($(Tx1-T1) + (-0.4, {0*\h})$) -- ($(Tx1-T1) + (-0.4, {-1/6*\h})$) node [left, pos=0.5] {\footnotesize$1/2$};
			\draw [<->] ($(Tx1-T1) + (-0.4, {-1/6*\h})$) -- ($(Tx1-T1) + (-0.4, {-1*\h})$) node [left, pos=0.5] {\footnotesize$5/2$};
			\draw [<->] ($(Tx2-T1) + (-0.4, {0*\h})$) -- ($(Tx2-T1) + (-0.4, {-1/3*\h})$) node [left, pos=0.5] {\footnotesize$1$};
			\draw [<->] ($(Tx2-T1) + (-0.4, {-1/3*\h})$) -- ($(Tx2-T1) + (-0.4, {-2/3*\h})$) node [left, pos=0.5] {\footnotesize$1$};
			\draw [<->] ($(Tx3-T1) + (-0.4, {0*\h})$) -- ($(Tx3-T1) + (-0.4, {-1/3*\h})$) node [left, pos=0.5] {\footnotesize$1$};
			\draw [<->] ($(Tx3-T1) + (-0.4, {-1/3*\h})$) -- ($(Tx3-T1) + (-0.4, {-1*\h})$) node [left, pos=0.5] {\footnotesize$2$};
			\draw [<->] ($(Tx2-T2) + (-0.4, {0*\h})$) -- ($(Tx2-T2) + (-0.4, {-1/6*\h})$) node [left, pos=0.5] {\footnotesize$1/2$};
			\draw [<->] ($(Tx3-T2) + (-0.4, {0*\h})$) -- ($(Tx3-T2) + (-0.4, {-1/6*\h})$) node [left, pos=0.5] {\footnotesize$1/2$};

			\foreach \t in {1,2}{
				\foreach \v in {1,2,3}{
					\draw ($(Rx\v-T\t) + (0, 0)$) -- ($(Rx\v-T\t) + (1, 0 )$); 
				}
			}
			\node at ($(Rx3-T1) + (0.5, {-0.45*\d})$) [anchor = south] {$\boldsymbol{r}$};
			\node at ($(Rx3-T2) + (0.5, {-0.45*\d})$) [anchor = south] {$\boldsymbol{r}'$};
			
			\foreach \t in {1,2}{
				\foreach \j in {1,2,3}{
					\coordinate (Rx1-T\t-\j) at ($(Rx1-T\t) + ({ \j*0.1+ (\j-1)*0.2}, 0)$);
				}
				\draw ($(Rx1-T\t-1)$) rectangle ($(Rx1-T\t-1) + (0.2, {1*\h})$);
				\draw ($(Rx1-T\t-2)$) rectangle ($(Rx1-T\t-2) + (0.2, {1/3*\h})$);
				\draw ($(Rx1-T\t-3)$) rectangle ($(Rx1-T\t-3) + (0.2, {1/6*\h})$);
				
			}
			\foreach \t in {1,2}{
				\foreach \j in {1,2,3}{
					\coordinate (Rx2-T\t-\j) at ($(Rx2-T\t) + ({ \j*0.1+ (\j-1)*0.2}, 0)$);
				}
				\draw ($(Rx2-T\t-1)$) rectangle ($(Rx2-T\t-1) + (0.2, {1/2*\h})$);
				\draw ($(Rx2-T\t-2)$) rectangle ($(Rx2-T\t-2) + (0.2, {1*\h})$);
				\draw ($(Rx2-T\t-3)$) rectangle ($(Rx2-T\t-3) + (0.2, {2/3*\h})$);
			}
			\foreach \t in {1,2}{
				\foreach \j in {1,2,3}{
					\coordinate (Rx3-T\t-\j) at ($(Rx3-T\t) + ({ \j*0.1+ (\j-1)*0.2}, 0)$);
				}
				\draw ($(Rx3-T\t-1)$) rectangle ($(Rx3-T\t-1) + (0.2, {1/6*\h})$);
				\draw ($(Rx3-T\t-2)$) rectangle ($(Rx3-T\t-2) + (0.2, {1/6*\h})$);
				\draw ($(Rx3-T\t-3)$) rectangle ($(Rx3-T\t-3) + (0.2, {1*\h})$);
			}
			
			\draw [fill = blue!50] ($(Rx1-T1-1) + (0, {0*\h})$) rectangle ($(Rx1-T1-1) + (0.2, {5/6*\h})$) node [pos=0.5] {\footnotesize$X_1$};
			
			\draw [fill = blue!50] ($(Rx2-T1-1) + (0, {0*\h})$) rectangle ($(Rx2-T1-1) + (0.2, {1/3*\h})$) node [pos=0.5] {\footnotesize$X_1$};
			\draw [fill = blue!50] ($(Rx2-T1-2) + (0, {1/3*\h})$) rectangle ($(Rx2-T1-2) + (0.2, {2/3*\h})$) node [pos=0.5] {\footnotesize$X_2$};
			\draw [fill = blue!50] ($(Rx2-T1-3) + (0, {0*\h})$) rectangle ($(Rx2-T1-3) + (0.2, {1/3*\h})$) node [pos=0.5] {\footnotesize $X_3$};

			\draw [fill = blue!50] ($(Rx3-T1-3) + (0, {0*\h})$) rectangle ($(Rx3-T1-3) + (0.2, {2/3*\h})$) node [pos=0.5] {\footnotesize $X_3$};

			\draw [fill = red!50] ($(Rx1-T2-1) + (0, {1*\h})$) rectangle ($(Rx1-T2-1) + (0.2, {1/6*\h})$) node [pos=0.5] {\footnotesize$X_1$};
			\draw [fill = red!50] ($(Rx1-T2-2) + (0, {0*\h})$) rectangle ($(Rx1-T2-2) + (0.2, {1/6*\h})$) node [pos=0.5] {\tiny$X_2$};

			\draw [fill = red!50] ($(Rx2-T2-1) + (0, {0*\h})$) rectangle ($(Rx2-T2-1) + (0.2, {1/2*\h})$) node [pos=0.5] {\footnotesize $X_1$};
			\draw [fill = red!50] ($(Rx2-T2-2) + (0, {1/2*\h})$) rectangle ($(Rx2-T2-2) + (0.2, {5/6*\h})$) node [pos=0.5] {\footnotesize$X_2$};
			\draw [fill = red!50] ($(Rx2-T2-3) + (0, {0*\h})$) rectangle ($(Rx2-T2-3) + (0.2, {1/2*\h})$) node [pos=0.5] {\footnotesize $X_3$};

			\draw [fill = red!50] ($(Rx3-T2-1) + (0, {0*\h})$) rectangle ($(Rx3-T2-1) + (0.2, {1/6*\h})$) node [pos=0.5] {\tiny$X_1$};
			\draw [fill = red!50] ($(Rx3-T2-3) + (0, {1/6*\h})$) rectangle ($(Rx3-T2-3) + (0.2, {5/6*\h})$) node [pos=0.5] {\footnotesize $X_3$};

		\end{tikzpicture}
	}
	\caption{\small \it An illustration of Lemma \ref{lemma:pwr}.  The power allocation  $\boldsymbol{r}'$ is obtained from $\boldsymbol{r}$ by raising the power levels at all inputs by $-r_{\max}=0.5$. Both allocations achieve the same GDoF tuple $(d_1,d_2,d_3)=(2.5,1,2)$.}
	\label{fig:pwr}
\end{figure}

Although Lemma \ref{lemma:mubound} gives an upper bound for the extremal GDoF gain, it remains to explicitly bound the sum-GDoF gain in a topology subset $\mathcal{A}_K^+$. 
The following lemma serves as a key building block for bounding such a gain.

\begin{lemma} \label{lemma:bo}
	Consider any network topology  $\tilde\balpha = [\tilde{\alpha}_{ij}]_{i \in [K], j \in [K]} \in\mathcal{A}_K^+$, and any power allocation vector $\tilde\bbr=(\tilde{r}_1,\tilde{r}_2,\cdots,\tilde{r}_K)\in\Omega^K$ 
	that achieves $\dsumo(\tilde\balpha)$, i.e.,
	\begin{align}
		\dsumo(\tilde\balpha)&=\sum_{k=1}^K d_k(\tilde\balpha,\tilde\bbr).
	\end{align}
	For any $\mathcal{S}=\{\beta_1, \beta_2, \cdots, \beta_L\} \subseteq [K]$ satisfying $\tilde{r}_{\beta_1} \geq \tilde{r}_{\beta_2} \geq \cdots \geq \tilde{r}_{\beta_L}$, we have 
	\begin{align} 
		\dsumb(\tilde\balpha) &\geq \sum_{i=1}^{L-1} \left( d_{\beta_i}(\tilde\balpha,\tilde\bbr)  + \tilde{r}_{\beta_{L}} - \tilde{r}_{\beta_{i}}   \right) + \left(d_{\beta_L}(\tilde\balpha,\tilde\bbr) + \tilde{r}_{\beta_{L-1}} - \tilde{r}_{\beta_{L}}  \right) ,&&\text{if } L>1,\label{eq:boundB}\\
		\dsumb(\tilde\balpha) &\geq d_{\beta_1}(\tilde\balpha, \tilde\bbr) - \tilde{r}_{\beta_1},&& \text{if } L=1. \label{eq:boundB-single}
	\end{align}
\end{lemma}
\begin{remark}For brevity, we will refer to the bounds (\ref{eq:boundB}), (\ref{eq:boundB-single}) as $B_{[ \beta_1, \beta_2, \cdots, \beta_L ]}$. 
\end{remark}
\begin{proof}
	First we note that, since $\tilde\balpha\in\mathcal{A}_K^+$, by (\ref{eq:tinGDoF})
	\begin{align} \label{eq:dio}
		{d}_{i}(\tilde\balpha,\tilde\bbr) = \tilde{\alpha}_{ii} + \tilde{r}_i - \max_{k \in [K], k\neq i} (\tilde{\alpha}_{ik}+ \tilde{r}_k)^+ > 0
	\end{align}	
 	for all $i\in[K]$.
	For any $\mathcal{S}\subseteq [K]$, define $\bbr(\mathcal{S})=(r_1(\mathcal{S}), r_2(\mathcal{S}), \cdots, r_K(\mathcal{S}))\in\Omega^K_b$ as the power allocation that sets $r_i(\mathcal{S})=0$ if $i\in\mathcal{S}$ and $r_i(\mathcal{S})=-\infty$ if $i\notin\mathcal{S}$. In other words, the power allocation $\bbr(\mathcal{S})$ amounts to switching on all transmitters in $\mathcal{S}$ at full power, and switching off all transmitters not in $\mathcal{S}$.
	We note that 
	\begin{align}
		\dsumb(\tilde\balpha) = \max_{\mathcal{S}\subseteq [K]} D_\Sigma(\tilde\balpha,\bbr(\mathcal{S})) \geq D_\Sigma(\tilde\balpha,\bbr(\mathcal{S})) \label{eq:dsumb_lbd}
	\end{align}
	for all $\mathcal{S}\subseteq [K].$
	If $L=1$, i.e., $\mathcal{S} = \{\beta_1\}  \subseteq [K]$, then the bound (\ref{eq:boundB-single}) is obtained from (\ref{eq:dsumb_lbd}) as follows:
	\begin{align}
		\dsumb(\tilde\balpha) &\geq D_\Sigma(\tilde\balpha, \bbr(\mathcal{S})) = \tilde\alpha_{\beta_1 \beta_1} \\
		&\geq \left(\tilde\alpha_{\beta_1 \beta_1} + \tilde{r}_{\beta_1} - \max_{k\in [K], k\neq \beta_1} (\tilde\alpha_{\beta_1 k} + \tilde{r}_k)^+ \right ) - \tilde{r}_{\beta_1} \\
		&= d_{\beta_1}(\tilde\balpha, \tilde\bbr) - \tilde{r}_{\beta_1}.
	\end{align}  
	
	Next we consider $L>1$, and prove the bound (\ref{eq:boundB}). We will create a new topology $\balpha'$ based on the given topology $\tilde\balpha$, such that the sum GDoF under OPC remains the same, i.e., $\dsumo(\tilde\balpha)=\dsumo(\balpha')$, while the sum GDoF under BPC cannot be better but may be possibly reduced, i.e., $\dsumb(\tilde\balpha)\geq\dsumb(\balpha')$. The main idea is that given a power allocation $\tilde\bbr$, since the GDoF achieved by each user are limited only by the strongest interference signal seen by its receiver, if the cross-channel strengths of the remaining interferers to that receiver are increased so that every interfering signal is as strong as the original strongest interference signal, this does not hurt the GDoF achieved under OPC, but it can potentially reduce  the GDoF achieved with BPC.

	For any given subset of users $\bbeta= \{\beta_1, \beta_2, \cdots, \beta_L\}$ and the given topology $\tilde\balpha$, we define the new topology $\balpha' = [\alpha_{ij}']_{i \in [K], j \in [K]} $, such that for all $i,j \in [K]$, 
	\begin{align} \label{eq:newAlpha}
		\alpha_{ij}' \triangleq \begin{cases}
			\tilde{\alpha}_{ij}, & \text{if } (i=j) \lor (i\notin\bbeta) \lor (j\notin\bbeta),\\ 
			\max_{k \in \bbeta, k\neq i} (\tilde{\alpha}_{ik} +\tilde{r}_{k})^+ - \tilde{r}_{j}, & \text{otherwise.}
		\end{cases}
	\end{align} 
	Note that $\alpha_{ij}' \geq \tilde{\alpha}_{ij}$, since for $i,j \in \bbeta$ and $i\neq j$, we have $\alpha_{ij}' \geq (\tilde{\alpha}_{ij} + \tilde{r}_{j})^+ - \tilde{r}_j \geq \tilde{\alpha}_{ij} + \tilde{r}_{j} - \tilde{r}_j = \tilde{\alpha}_{ij}$. Also note that relative to $\tilde\balpha$ the new topology $\balpha'$ differs in only those cross-links that both emerge and terminate at users in $\bbeta$. For such cross links, their strengths are lifted up to the extent that, with the power allocations $\tilde\bbr$, their interference levels at each receiver reach the maximal interference level at that receiver due to all undesired transmitters in $\bbeta$ in the original topology $\tilde\balpha$.
	Figure \ref{fig:newNet} illustrates an example of how the new topology $\balpha'$ is created from $\tbalpha$, where we assume $\tilde{r}_1 \geq \tilde{r}_2 \geq \cdots \geq \tilde{r}_K$, choose $\bbeta = [K]$ and focus on Receiver 1. We firstly identify the maximal interference level in the original topology in Figure \ref{fig:newNet}(a) as $\max_{k\in \bbeta, k \neq 1} (\tilde{\alpha}_{1k} +\tilde{r}_{k})^+ = \tilde{\alpha}_{12} +\tilde{r}_{2}$. Then in the new topology, we define $\alpha_{11}' = \tilde{\alpha}_{11}$, and $\alpha_{1j}' = \tilde{\alpha}_{12} +\tilde{r}_{2} - \tilde{r}_{j}$ for $j = 2,3,\cdots,K$. In the new topology $\balpha'$ with the power allocations $\tbbr$, the interference caused at Receiver $1$ by all undesired transmitters is thus at the same level, as shown in Figure \ref{fig:newNet}(b).
	In the new topology $\balpha'$, the sum GDoF $\dsumo(\balpha')\geq \dsumo(\tilde\balpha)$ because the GDoF achieved by the original power allocation $\tilde\bbr$ remain unchanged in the new topology. On the other hand, because increasing the cross-channel strengths cannot help a TIN scheme under OPC, this implies that $\dsumo(\balpha')= \dsumo(\tilde\balpha)$. For the same reason, i.e., because increasing cross-channel strengths cannot help a TIN scheme under BPC, we have $\dsumb(\tilde\balpha) \geq \dsumb(\balpha')$.
	
	Next, in the new topology $\balpha'$ let us apply the binary power allocation vector $\bbr(\bbeta)$, i.e., switch on the transmitters in $\bbeta$ at full power and keep the remaining transmitters switched off. Obviously, $\dsumb(\balpha') \geq D_\Sigma(\balpha',\bbr(\bbeta)) $. 
	
	Next let us show that for $i \in \bbeta$, we can further bound $d_i(\balpha', \bbr(\bbeta))$ from below as
	\begin{align}
		d_i(\balpha', \bbr(\bbeta)) &\geq {d}_{i}(\tilde\balpha,\tilde\bbr) - \tilde{r}_i + \min_{j \in \bbeta, j \neq i} \tilde{r}_{j}.\label{eq:newdb}
	\end{align} 
	This is because, by the definition of $d_i(\balpha', \bbr(\bbeta))$, we have for $i\in\bbeta$,
	\begin{align}
		d_i(\balpha', \bbr(\bbeta))
		&= \left( \tilde{\alpha}_{ii} - \max_{j \in \bbeta, j \neq i} \alpha_{ij}'  \right)^+ \label{eq:db-1} \\
		&\geq \tilde{\alpha}_{ii} - \max_{j \in \bbeta, j \neq i} \alpha_{ij}' \label{eq:db-2}\\
		&= \tilde{\alpha}_{ii} - \max_{j \in \bbeta, j \neq i} \left \{ \max_{k \in \bbeta, k \neq i} (\tilde{\alpha}_{ik} + \tilde{r}_{k})^+ - \tilde{r}_j \right \} \label{eq:db-3}\\
		&= \tilde{\alpha}_{ii} - \max_{k \in \bbeta, k \neq i} (\tilde{\alpha}_{ik} + \tilde{r}_{k})^+ +  \min_{j \in \bbeta, j \neq i} \tilde{r}_j \label{eq:db-4}\\
		&=  \left( (\tilde{\alpha}_{ii} + \tilde{r}_{i}) - \max_{k \in \bbeta, k \neq i} (\tilde{\alpha}_{ik} + \tilde{r}_{k})^+ \right ) - \tilde{r}_{i} +  \min_{j \in \bbeta, j \neq i} \tilde{r}_j \label{eq:db-5}\\
		&\geq  \left( (\tilde{\alpha}_{ii} + \tilde{r}_{i}) - \max_{k \in [K], k \neq i} (\tilde{\alpha}_{ik} + \tilde{r}_{k})^+ \right ) - \tilde{r}_{i} + \min_{j \in \bbeta, j \neq i} \tilde{r}_j \label{eq:db-6}\\
		&= {d}_{i}(\tilde\balpha,\tilde\bbr) - \tilde{r}_{i}  +  \min_{j \in \bbeta, j \neq i} \tilde{r}_j.\label{eq:db-7}
	\end{align}
	Equality (\ref{eq:db-3}) follows from the definition of $\alpha_{ij}'$  in (\ref{eq:newAlpha}). In (\ref{eq:db-3}) we note that the inner maximum is a constant with respect to the outer one; thus we have (\ref{eq:db-4}).
	Step (\ref{eq:db-6}) holds because we expand the scope of values over which the first maximum is taken, from $\bbeta$ to $[K]$. 
	Finally,  (\ref{eq:db-7}) follows due to (\ref{eq:dio}).

	Summing up (\ref{eq:newdb}) for all $i \in \bbeta$ and applying  the assumption $\tilde{r}_{\beta_1} \geq \tilde{r}_{\beta_2} \geq \cdots \geq \tilde{r}_{\beta_L}$, we obtain
	\begin{align} 
		D_\Sigma(\balpha',\bbr(\bbeta)) \geq \sum_{i= 1}^{L-1} \left(\tilde{d}_{\beta_i}(\tilde\balpha,\tilde\bbr) - \tilde{r}_{\beta_i} + \tilde{r}_{\beta_L}\right) + \left(\tilde{d}_{\beta_L}(\tilde\balpha,\tilde\bbr) - \tilde{r}_{\beta_L} + \tilde{r}_{\beta_{L-1}}\right).\label{eq:sumNewdb}
	\end{align}
	By putting (\ref{eq:sumNewdb}) together with $\dsumb(\tilde\balpha)\geq\dsumb(\balpha')  \geq D_\Sigma(\balpha',\bbr(\bbeta))$, we arrive at the bound  in (\ref{eq:boundB}).
\end{proof}

\begin{figure}[!t] 
	\centering
	\begin{subfigure}[b]{0.57\textwidth}
		\resizebox{\textwidth}{!}{
			\begin{tikzpicture}[scale=1]
				\def \w {1}
				\def \d {4}
				\def \h {1}
				\def \e {0.1}
				\node (gnd1) at (0,0) {};
				\node (gnd2) at (0, 0) {};
				
				\draw [line width = 0.5] (gnd1) -- ($(gnd1)  + ({7*\w}, 0)$);
				
				\draw [fill=white] ($(gnd1) + ({0.5*\w},0)$) rectangle ($(gnd1) + ({1.5*\w},{2*\h})$);
				\draw [fill=white] ($(gnd1) + ({2*\w},0)$) rectangle ($(gnd1) + ({3*\w},{1*\h})$);
				\draw [fill=white] ($(gnd1) + ({3.5*\w},0)$) rectangle ($(gnd1) + ({4.5*\w},{0.5*\h})$);
				\node at ($(gnd1) + ({5*\w}, \h) $) {$\cdots$};
				\draw [fill=white] ($(gnd1) + ({5.5*\w},0)$) rectangle ($(gnd1) + ({6.5*\w},{1.5*\h})$);
				
				\draw [fill=blue!50] ($(gnd1) + ({0.5*\w},0)$) rectangle ($(gnd1) + ({1.5*\w},{1.8*\h})$);
				\draw [fill=red!50] ($(gnd1) + ({2*\w},0)$) rectangle ($(gnd1) + ({3*\w},{0.8*\h})$);
				\draw [fill=red!50] ($(gnd1) + ({3.5*\w},0)$) rectangle ($(gnd1) + ({4.5*\w},{0.2*\h})$);
				\draw [fill=red!50] ($(gnd1) + ({5.5*\w},0)$) rectangle ($(gnd1) + ({6.5*\w},{0.2*\h})$);
				
				\draw [line width = 0.5, <->] ($(gnd1) + ({1.5*\w+\e},{1.8*\h})$) -- ($(gnd1)  + ({1.5*\w+\e},{2*\h})$) node [right = -2.5, pos=0.5] {\tiny$|\tilde{r}_1|$};
				\draw [line width = 0.5, <->] ($(gnd1) + ({3*\w+\e},{0.8*\h})$) -- ($(gnd1)  + ({3*\w+\e},{1*\h})$) node [right = -2.5, pos=0.5] {\tiny$|\tilde{r}_2|$};
				\draw [line width = 0.5, <->] ($(gnd1) + ({4.5*\w+\e},{0.2*\h})$) -- ($(gnd1)  + ({4.5*\w+\e},{0.5*\h})$) node [right = -2.5, pos=0.5] {\tiny$|\tilde{r}_3|$};
				\draw [line width = 0.5, <->] ($(gnd1) + ({6.5*\w+\e},{0.2*\h})$) -- ($(gnd1)  + ({6.5*\w+\e},{1.5*\h})$) node [right = -2.5, pos=0.75] {\tiny$|\tilde{r}_K|$};
				\node at ($(gnd1) + ({1.5*\w},{2*\h})$) [above right] {\small $\tilde{\alpha}_{11}$};
				\node at ($(gnd1) + ({3*\w},{1*\h})$) [above right] {\small $\tilde{\alpha}_{12}$};
				\node at ($(gnd1) + ({4.5*\w},{0.5*\h})$) [above = 3, right] {\small $\tilde{\alpha}_{13}$};
				\node at ($(gnd1) + ({6.5*\w},{1.5*\h})$) [above right] {\small $\tilde{\alpha}_{1K}$};
				\draw [line width = 0.5, dashed] ($(gnd1) + (0,{0.8*\h})$) node [left, align = center] {Interference\\Level}-- ($(gnd1) + ({7*\w},{0.8*\h})$);
				
				\node at ($(gnd1) + (\w, -0.4)$) {\small$\sqrt{P^{\tilde{\alpha}_{11}}} X_1$};
				\node at ($(gnd1) + ({2.5*\w}, -0.4)$) {\small$\sqrt{P^{\tilde{\alpha}_{12}}} X_2$};
				\node at ($(gnd1) + ({4*\w}, -0.4)$) {\small$\sqrt{P^{\tilde{\alpha}_{13}}} X_3$};
				\node at ($(gnd1) + ({6*\w}, -0.4)$) {\small$\sqrt{P^{\tilde{\alpha}_{1K}}} X_K$};
				
			\end{tikzpicture}			
		}
		\caption{}
	\end{subfigure}\hfill
	\begin{subfigure}[b]{0.43\textwidth}
		\resizebox{\textwidth}{!}{
			\begin{tikzpicture}[scale=1]
				
				\def \w {1}
				\def \d {4}
				\def \h {1}
				\def \e {0.1}
				\node (gnd1) at (0,0) {};
				\node (gnd2) at (0,0) {};
				
				\draw [line width = 0.5] (gnd2) -- ($(gnd2)  + ({7*\w}, 0)$);
				
				\draw [fill=white] ($(gnd2) + ({0.5*\w},0)$) rectangle ($(gnd2) + ({1.5*\w},{2*\h})$);
				\draw [fill=white] ($(gnd2) + ({2*\w},0)$) rectangle ($(gnd2) + ({3*\w},{1*\h})$);
				\draw [fill=white] ($(gnd2) + ({3.5*\w},0)$) rectangle ($(gnd2) + ({4.5*\w},{1.1*\h})$);
				\node at ($(gnd2) + ({5*\w}, 0.5*\h) $) {$\cdots$};
				\draw [fill=white] ($(gnd2) + ({5.5*\w},0)$) rectangle ($(gnd2) + ({6.5*\w},{2.1*\h})$);
				
				\draw [fill=blue!50] ($(gnd2) + ({0.5*\w},0)$) rectangle ($(gnd2) + ({1.5*\w},{1.8*\h})$);
				\draw [fill=red!50] ($(gnd2) + ({2*\w},0)$) rectangle ($(gnd2) + ({3*\w},{0.8*\h})$);
				\draw [fill=red!50] ($(gnd2) + ({3.5*\w},0)$) rectangle ($(gnd2) + ({4.5*\w},{0.8*\h})$);
				\draw [fill=red!50] ($(gnd2) + ({5.5*\w},0)$) rectangle ($(gnd2) + ({6.5*\w},{0.8*\h})$);
				
				\draw [line width = 0.5, <->] ($(gnd2) + ({1.5*\w+\e},{1.8*\h})$) -- ($(gnd2)  + ({1.5*\w+\e},{2*\h})$) node [right = -2.5, pos=0.5] {\tiny$|\tilde{r}_1|$};
				\draw [line width = 0.5, <->] ($(gnd2) + ({3*\w+\e},{0.8*\h})$) -- ($(gnd2)  + ({3*\w+\e},{1*\h})$) node [right = -2.5, pos=0.5] {\tiny$|\tilde{r}_2|$};
				\draw [line width = 0.5, <->] ($(gnd2) + ({4.5*\w+\e},{0.8*\h})$) -- ($(gnd2)  + ({4.5*\w+\e},{1.1*\h})$) node [right = -2.5, pos=0.5] {\tiny$|\tilde{r}_3|$};
				\draw [line width = 0.5, <->] ($(gnd2) + ({6.5*\w+\e},{0.8*\h})$) -- ($(gnd2)  + ({6.5*\w+\e},{2.1*\h})$) node [right = -2.5, pos=0.5] {\tiny$|\tilde{r}_K|$};
				\node at ($(gnd2) + ({1.5*\w},{2*\h})$) [above right] {\small $\alpha'_{11} = \tilde{\alpha}_{11}$};
				\node at ($(gnd2) + ({3*\w},{1.1*\h})$) [above right] {\small $\alpha'_{12}$};
				\node at ($(gnd2) + ({4.5*\w},{1.1*\h})$) [above right] {\small $\alpha'_{13}$};
				\node at ($(gnd2) + ({6.5*\w},{2.1*\h})$) [above right] {\small $\alpha'_{1K}$};
				\draw [line width = 0.5, dashed] ($(gnd2) + (0,{0.8*\h})$) -- ($(gnd2) + ({7*\w},{0.8*\h})$);
				
				\node at ($(gnd2) + (\w, -0.4)$) {\small$\sqrt{P^{\alpha'_{11}}} X_1$};
				\node at ($(gnd2) + ({2.5*\w}, -0.4)$) {\small$\sqrt{P^{\alpha'_{12}}} X_2$};
				\node at ($(gnd2) + ({4*\w}, -0.4)$) {\small$\sqrt{P^{\alpha'_{13}}} X_3$};
				\node at ($(gnd2) + ({6*\w}, -0.4)$) {\small$\sqrt{P^{\alpha'_{1K}}} X_K$};
				
			\end{tikzpicture}
		}
		\caption{}
	\end{subfigure}
	
	\caption{\small \it Creating a new network topology by raising the strength levels of the cross links in the original one. The blue bars are the desired signals, while the red bars are the interfering signals. The heights of the colored bars denote their power levels in dB scale. (a) The signals at Receiver 1 under optimal power control in the original network. (b) The signals with the same power allocation in the new network.} 
	\label{fig:newNet}
\end{figure}
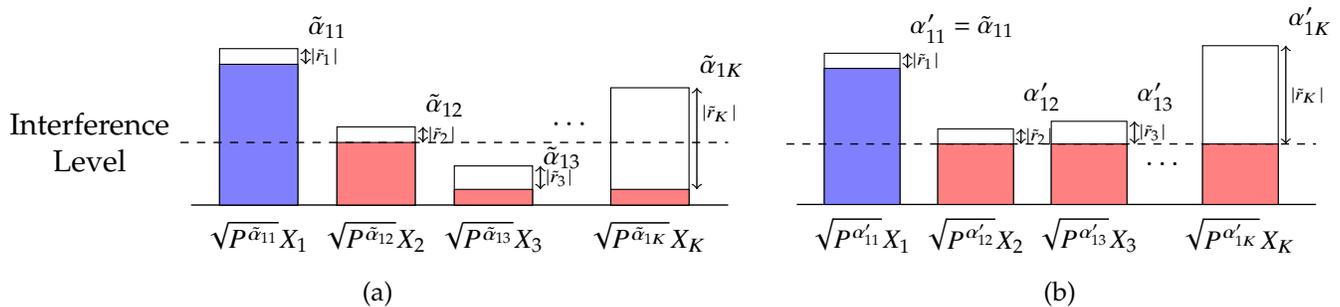

The next lemma is essential to this work as it bounds  the sum-GDoF gains as a function of $K$. 

\begin{lemma} \label{lemma:K^2bounds}
	Let $m \in \mathbb{N}$. When $K \geq m^2$, there exists a topology $\balpha^* \in \mathcal{A}_K$ such that
	\begin{align}
		\frac{\dsumo(\balpha^*)}{\dsumb(\balpha^*)} & \geq m \label{eq:K^2lbd}.
	\end{align} 
	On the other hand, when $K = m^2$, we have
	\begin{align}
		\sup_{\balpha \in \mathcal{A}_K^+} \frac{\dsumo(\balpha)}{\dsumb(\balpha)} &\leq \frac{3}{2}m-\frac{1}{2}. \label{eq:K^2ubd} 
	\end{align}
\end{lemma}
\begin{proof}
	First we show the lower bound (\ref{eq:K^2lbd}) with a topology $\balpha^* \in \mathcal{A}_{m^2}$ depicted in Figure \ref{fig:extremalK}.
	We divide the $m^2$ users into $m$ groups, each of which consists of $m$ users and is denoted by $\mathcal{G}_i$, $i\in\{1,2,...,m\}$, i.e.,
	\begin{align}
		\mathcal{G}_i=\left\{(i-1)m+1,(i-1)m+2,...,im \right\}.
	\end{align}
	In each group $\mathcal{G}_i$, all direct links have strength equal to $i$.
	For cross links, each transmitter in the group $\mathcal{G}_i$  connects to each receiver in $\mathcal{G}_j$ with strength $i-1$, where $j\leq i$ and $i,j\in\{1,2,...,m\}$. 
	All the other cross links are disconnected, i.e., transmitters in the $i^{th}$ group do not interfere with receivers in groups $i+1, i+2,\cdots, m$.
	As an example, when $m=2$, we get the $4$-user interference channel in Figure \ref{fig:extremalnet}(b), where $\mathcal{G}_1=\{1,2\}$ and $\mathcal{G}_2=\{3,4\}$.

	For the topology $\balpha^*$, each user achieves one GDoF when the power allocation vector $\bbr^* = \left(r_1^*, r_2^*, \cdots, r_K^* \right)$ is set as 
	\begin{align}
		r_k^*&=1-j, && \forall k\in\mathcal{G}_j,j\in [m]. \label{eq:KKpwr}
	\end{align}
	Thus the sum GDoF $D_\Sigma(\balpha^*, \bbr^*) = m^2$, implying that $\dsumo(\balpha^*) \geq m^2$. 
	Now consider binary power control, where we either activate a user by switching his transmitter on at full power, or deactivate the user by switching his transmitter off.
	In the topology $\balpha^*$, activating any one of the transmitters in group $\mathcal{G}_i$ at full power overwhelms all the receivers in $\mathcal{G}_j$, for all $j < i$; i.e., all users in $\bigcup_{j \in [K], j < i}\mathcal{G}_j$ get zero GDoF whether or not they are activated. This means that only the users in the group with the largest index will survive. So it suffices to consider one single group at a time and decide how many users in that group may be activated to maximize the sum GDoF. Activating one user in group $i$ yields sum-GDoF value of $i$, and activating $k$ users, $1<k\leq m$, in group $i$ yields a sum-GDoF value of $k$.  The optimal value is therefore $\dsumb(\balpha^*) = m$, which is achieved by activating all $m$ users in one single group only, or by activating any one user in group $\mathcal{G}_m$.
	Therefore, we establish the lower bound $\frac{\dsumo(\balpha^*)}{\dsumb(\balpha^*)} \geq m$ for $K = m^2$. 
	For $K \geq m^2$, the same lower bound still holds, because we can achieve the same bound with a topology $\balpha^* \in \mathcal{A}_K$, where we copy the one depicted in Figure \ref{fig:extremalK} to the first $m^2$ users, and remove all links associated with the rest $K-m^2$ users. Here we conclude the proof of (\ref{eq:K^2lbd}).

	Next, we show the upper bound (\ref{eq:K^2ubd}) for any $\tbalpha \in \mathcal{A}_K^+$ when $K = m^2$. Let $\tbbr$ be a power allocation vector associated with $\tbalpha$ as in Lemma \ref{lemma:bo}. 
	A key non-trivial step of the proof is the idea of partitioning the sum-GDoF in a particular way and applying bounds obtained from Lemma \ref{lemma:bo} to each partition. In the following we first present the bounds, and then  the explicit sum-GDoF partition.
	Without loss of generality, let us assume that
		\begin{align}
			\tilde{r}_1\geq \tilde{r}_2\geq\cdots\geq \tilde{r}_K,\label{eq:assumedelta}
		\end{align}
	which carries no loss of generality because it only amounts to a labeling of users 
	in this order.

	First consider the bounds $B_{[i,i+1,...i+m-1]}$. With the assumption (\ref{eq:assumedelta}), 
	\begin{align}
		B_{[i,i+1,...i+m-1]} \implies \dsumb(\tbalpha) \geq \sum_{j=i}^{i+m-1} d_j(\tbalpha, \tbbr) - \sum_{j=i}^{i+m-3} \tilde{r}_j + (m-2) \tilde{r}_{i+m-1}, \label{eq:B1}
	\end{align}
	where $i \in [m^2-m+1]$. Adding  the bounds $\sum_{i=1}^{m^2-m+1}B_{[i,i+1,...i+m-1]}$ we obtain
	\begin{align}
		&(m^2 -m + 1) \dsumb(\tbalpha) \notag \\
		&\geq \sum_{j=1}^{m-1} j d_j(\tbalpha, \tbbr)- \sum_{j=1}^{m-2} j \tilde{r}_j - (m-2) \tilde{r}_{m-1} +  \sum_{j=m}^{m^2-m+1} m d_j(\tbalpha, \tbbr) \notag \\
		&\quad + \sum_{j=1}^{m-1} (m-j) d_{m^2-m+1+j}(\tbalpha, \tbbr) + \sum_{j=1}^{m-2} j \tilde{r}_{m^2-m+1+j} + (m-2) \tilde{r}_{m^2}. \label{eq:B1sum}
	\end{align}
	Then we consider the bounds $B_{[1,2,...,i]}$, which are obtained from Lemma \ref{lemma:bo} and the assumption (\ref{eq:assumedelta}):
	\begin{align}
		B_{[1]} &\implies \dsumb(\tbalpha) \geq d_1(\tbalpha, \tbbr) - \tilde{r}_1, \label{eq:B21}\\
		B_{[1,2,...,i]} &\implies \dsumb(\tbalpha) \geq \sum_{j=1}^{i} d_j(\tbalpha, \tbbr) - \sum_{j=1}^{i-2} \tilde{r}_j + (i-2) \tilde{r}_{i}, \label{eq:B22}
	\end{align}
	where $i = 2,3,\cdots, m-1$. Adding  $B_{[1]} +\sum_{i=2}^{m-1}B_{[1,2,...,i]}$, we obtain,
	\begin{align}
		&(m-1) \dsumb(\tbalpha) \notag \\
		&\geq \sum_{j=1}^{m-1} (m-j) d_j(\tbalpha, \tbbr) - \tilde{r}_1 - \sum_{j=1}^{m-3}(m-2-j) \tilde{r}_j+ \sum_{j=3}^{m-1} (j-2) \tilde{r}_{j}\label{eq:B2sum-1}\\
		&\geq \sum_{j=1}^{m-1} (m-j) d_j(\tbalpha, \tbbr) + \sum_{j=3}^{m-1} (j-2) \tilde{r}_{j},\label{eq:B2sum-2}
	\end{align}
	where (\ref{eq:B2sum-2}) holds because $\tilde{r}_j\leq 0$ for $j \in [m^2]$.
	
	Next we consider the bounds $B_{[i+m^2-m+1]}$ together with (\ref{eq:assumedelta}), 
	\begin{align}
		B_{[i+m^2-m+1]} \implies \dsumb(\tbalpha) \geq d_{i+m^2-m+1}(\tbalpha, \tbbr) - \tilde{r}_{i+m^2-m+1}, \label{eq:B3}
	\end{align}
	where $i = 1,2,\cdots, m-1$. Adding $\sum_{i=1}^{m-1}iB_{[i+m^2-m+1]}$ we obtain,
	\begin{align}
		&\sum_{i=1}^{m-1} i \dsumb(\tbalpha) = \frac{1}{2}m(m-1) \dsumb(\tbalpha) \notag \\
		&\geq \sum_{i=1}^{m-1} \Big(i d_{i+m^2-m+1}(\tbalpha, \tbbr) - i \tilde{r}_{i+m^2-m+1}\Big). \label{eq:B3sum}
	\end{align}
	Finally, we sum up the bounds (\ref{eq:B1sum}), (\ref{eq:B2sum-2}) and (\ref{eq:B3sum}) and obtain
	\begin{align}
		&m\left(\frac{3}{2}m - \frac{1}{2} \right) \dsumb(\tbalpha) \notag \\
		&= \left[(m^2 -m + 1) + (m-1) + \frac{1}{2}m(m-1) \right] \dsumb(\tbalpha) \\
		& \geq  m \dsumo(\tbalpha) - \tilde{r}_1 - 2 \sum_{j=2}^{m-2} \tilde{r}_j -\tilde{r}_{m-1} - \tilde{r}_{m^2} \label{eq:BBsum-1}\\
		& \geq m \dsumo(\tbalpha),\label{eq:BBsum-2}
	\end{align}
	where (\ref{eq:BBsum-2}) holds because $\tilde{r_i} \leq 0 $ for all $i \in [m^2]$.
	Since inequality (\ref{eq:BBsum-2}) holds for all $\balpha \in \mathcal{A}_K^+$ regardless of the assumption (\ref{eq:assumedelta}), we have 
	\begin{align}
		\sup_{\balpha \in \mathcal{A}_K^+} \frac{\dsumo(\balpha)}{\dsumb(\balpha)} \leq \frac{3}{2}m-\frac{1}{2},
	\end{align}
	which concludes the proof of (\ref{eq:K^2ubd}).
\end{proof}

\section{Proof of Theorem \ref{thm:2-6usr}} \label{sec:proof2-6usr}

\subsection{Case: $K=2$}
It is known from \cite{2user_1,2user_2} already that binary power control is optimal for sum rate maximization in the $2$ user interference channel. 
This result applies to every SNR setting, and thus also to the GDoF framework, yielding $\mu_2=1$.
However, in order to lay the foundation of the proof for subsequent cases, let us  present an alternative proof based on Lemma \ref{lemma:bo}. Since OPC cannot be worse than BPC, the lower bound is trivial in this case, i.e., $\mu_2\geq 1$. 
Now consider the converse. For any network topology $\tilde\balpha\in\mathcal{A}_2^+$ and power allocation vector $\tilde\bbr$ as in Lemma \ref{lemma:bo}, we apply the bound $B_{[1,2]}$,
\begin{align}
	B_{[1,2]} &\implies \dsumb(\tilde\balpha) \geq (d_{1}(\tilde\balpha,\tilde\bbr)-\tilde{r}_1+\tilde{r}_2)+(d_{2}(\tilde\balpha,\tilde\bbr)-\tilde{r}_2+\tilde{r}_1) = \dsumo(\tilde\balpha). \label{eq:c2_1}
\end{align}
Therefore, we have $\sup_{\balpha \in \mathcal{A}_2^+} \frac{\dsumo(\balpha)}{\dsumb(\balpha)} \leq 1$, and Lemma \ref{lemma:mubound} yields the desired outer bound $\mu_2\leq \max \{ \mu_1, 1\} = 1$.  $\hfill\square$

\subsection{Case: $K=3$}
Let us start with the converse. 
For any network topology $\tilde\balpha \in \mathcal{A}_3^+$, and the power allocation vector $\tilde\bbr$ as in Lemma \ref{lemma:bo}, the bounds $B_{[1,2]}, B_{[2,3]}$ and $B_{[1,3]}$ follow the same reasoning in (\ref{eq:c2_1}) as follows:
\begin{align}
	B_{[1,2]} &\implies \dsumb(\tilde\balpha) \geq d_1(\tilde\balpha, \tilde\bbr) + d_2(\tilde\balpha, \tilde\bbr), \label{eq:c3_1}\\
	B_{[2,3]} &\implies \dsumb(\tilde\balpha) \geq d_2(\tilde\balpha, \tilde\bbr) + d_3(\tilde\balpha, \tilde\bbr), \label{eq:c3_2}\\
	B_{[1,3]} &\implies \dsumb(\tilde\balpha) \geq d_1(\tilde\balpha, \tilde\bbr) + d_3(\tilde\balpha, \tilde\bbr). \label{eq:c3_3}
\end{align}
Summing up (\ref{eq:c3_1})--(\ref{eq:c3_3}), we have $\sup_{\balpha \in \mathcal{A}_3^+} \frac{\dsumo(\balpha)}{\dsumb(\balpha)} \leq 3/2$. By Lemma \ref{lemma:mubound} we conclude the outer bound $\mu_3 \leq \max\{ \mu_2, 3/2 \} = 3/2.$

Next, for the lower bound on $\mu_3$, we consider the topology, say $\balpha_3$, depicted in Figure \ref{fig:extremalnet}(a). This topology has direct links set with strengths $\alpha_{11}=\alpha_{22}=1$  and $\alpha_{33} = 2$. The (red) cross links are set with strength 1.
With the power allocation $\bbr_3 = (0,0,-1)$, we  achieve the GDoF tuple $(1,1,1)$, hence  $\dsumo(\balpha_3)\geq 3$. 
It is not difficult to exhaustively verify that with binary power control $\dsumb(\balpha_3) = 2$.
Therefore, we establish the lower bound $\mu_3 \geq 3/2$. Since the lower bound matches the upper bound, we have $\mu_3=3/2$, which completes the proof for the case of $K=3$ users. $\hfill\square$

\subsection{Case: $K=4$}
Start with the converse. For any network topology in $\tilde\balpha \in \mathcal{A}_4^+$, and the power allocation vector $\tilde\bbr$ as in Lemma \ref{lemma:bo}, we have the following bounds, similar to (\ref{eq:c2_1}):
\begin{align}
	B_{[1,2]} &\implies \dsumb(\tilde\balpha) \geq d_1(\tilde\balpha, \tilde\bbr) + d_2(\tilde\balpha, \tilde\bbr), \label{eq:c4_1}\\
	B_{[3,4]} &\implies \dsumb(\tilde\balpha) \geq d_3(\tilde\balpha, \tilde\bbr) + d_4(\tilde\balpha, \tilde\bbr). \label{eq:c4_2}
\end{align}
Adding (\ref{eq:c4_1})--(\ref{eq:c4_2}) we obtain $\sup_{\balpha \in \mathcal{A}_4^+} \frac{\dsumo(\balpha)}{\dsumb(\balpha)} \leq 2$. Then we apply Lemma \ref{lemma:mubound} we get the desired upper bound $\mu_4\leq \max \{ \mu_3, 2 \} = 2$.

For the lower bound, we consider the topology, say $\balpha_4$, depicted in Figure \ref{fig:extremalK}(b), where the direct links have strengths $\alpha_{11}=\alpha_{22} = 1$ and $\alpha_{33} = \alpha_{44} = 2$, and the (red) cross links have strength 1.
We  achieve GDoF tuple $(1,1,1,1)$ with the power allocation $\bbr_4 =(0,0,-1,-1)$, so $\dsumo(\balpha_4) \geq 4$. For binary power control, by considering all $2^4$ possible power allocation vectors, we find $\dsumb(\balpha_4) = 2$.
Therefore, we establish the lower bound $\mu_4 \geq 2$. Since the lower bound matches the upper bound, we have $\mu_4=2$, which completes the proof for the case of $K=4$ users. $\hfill\square$

\subsection{Case: $K=5$} \label{section: K5}
Again, let us start with the converse. 
Consider any network topology $\tilde\balpha \in \mathcal{A}_5^+$, and any associated power allocation vector $\tbbr$ as in Lemma \ref{lemma:bo}. Without loss of generality, we further assume
\begin{align}
	0 \geq \tilde{r}_1 \geq \tilde{r}_2 \geq \cdots \geq \tilde{r}_5. \label{eq:assumedelta_K=5}
\end{align}
The assumption that $\tilde{r}_1 \geq \tilde{r}_2 \geq \cdots \geq \tilde{r}_5$ incurs no loss of generality since it is equivalent to the labeling of users in this order. The assumption $\tilde{r}_1=0$, which means that User $1$ transmits at full power, is justified by Lemma \ref{lemma:pwr}.
By Lemma \ref{lemma:bo} and the assumption (\ref{eq:assumedelta_K=5}), we have the following bounds.
\begin{align}
	B_{[5]}&\implies \dsumb(\tbalpha) \geq d_5(\tilde\balpha, \tilde\bbr)-\tilde{r}_5, \label{eq:c5_1}\\
	B_{[1,2]}&\implies \dsumb(\tbalpha) \geq d_1(\tilde\balpha, \tilde\bbr)+d_2(\tilde\balpha, \tilde\bbr), \label{eq:c5_2}\\
	B_{[1,4]}&\implies \dsumb(\tbalpha) \geq d_1(\tilde\balpha, \tilde\bbr)+d_3(\tilde\balpha, \tilde\bbr), \label{eq:c5_3}\\
	B_{[2,4]}&\implies \dsumb(\tbalpha) \geq d_2(\tilde\balpha, \tilde\bbr)+d_4(\tilde\balpha, \tilde\bbr), \label{eq:c5_4}\\
	B_{[1,2,3]}&\implies \dsumb(\tbalpha) \geq (d_1(\tilde\balpha, \tilde\bbr)+\tilde{r}_3)+(d_2(\tilde\balpha, \tilde\bbr)-\tilde{r}_2+\tilde{r}_3)+(d_3(\tilde\balpha, \tilde\bbr) - \tilde{r}_3+\tilde{r}_2), \label{eq:c5_5}\\
	B_{[3,4,5]}&\implies \dsumb(\tbalpha) \geq (d_3(\tilde\balpha, \tilde\bbr)-\tilde{r}_3+\tilde{r}_5)+(d_4(\tilde\balpha, \tilde\bbr)-\tilde{r}_4+\tilde{r}_5)+(d_5(\tilde\balpha, \tilde\bbr)-\tilde{r}_5+\tilde{r}_4). \label{eq:c5_6}
\end{align}
Summing up (\ref{eq:c5_1})--(\ref{eq:c5_6}), we obtain  
$\frac{\dsumo(\tbalpha)}{\dsumb(\tbalpha)} \leq 9/4$. This inequality holds for all $\balpha \in \mathcal{A}_5^+$ regardless of the assumption (\ref{eq:assumedelta_K=5}); as a result, we have $\sup_{\balpha \in \mathcal{A}_5^+} \frac{\dsumo(\balpha)}{\dsumb(\balpha)} \leq 9/4$. By Lemma \ref{lemma:mubound} we obtain the upper bound $\mu_5 \leq \max\{\mu_4, 9/4\} = 9/4$.

Next, we consider the topology, say $\balpha_5$, depicted in Figure \ref{fig:extremalnet}(c) for the lower bound on $\mu_5$. The direct links have strengths $\alpha_{ii} = 2$ for $i=1,2,3$, and $\alpha_{ii} = 4$ for $i = 4,5$. The strengths of the cross links are set as $\alpha_{i3} = 1$ for $i = 1,2,4,5$ (red solid), and $\alpha_{i4} = \alpha_{j5} = 2$ for $i, j \in [5], i\neq 4$ and $j \neq 5$ (blue dashed). With the power allocation vector $\bbr_5 =(0,0,-1,-2,-2)$, we  achieve the GDoF tuple $(2,2,1,2,2)$, so $\dsumo(\balpha_5) \geq 9$. On the other hand, we can also exhaustively verify that with binary power control, $\dsumb(\balpha_5)=4$. Therefore, we establish the lower bound $\mu_5\geq 9/4$. Since this matches the upper bound, we conclude that $\mu_5=9/4$, which completes the proof for the case of $K=5$ users. $\hfill\square$

\subsection{Case: $K=6$}
We start with the converse. For any network topology $\tilde\balpha \in \mathcal{A}_6^+$ and the associated $\tilde\bbr$ as defined in Lemma \ref{lemma:bo}, we first break down $16 \dsumo(\tilde\balpha)$ as follows.
\begin{align}
	&16 \dsumo(\tilde\balpha) \notag \\
	&= 6d_5(\tilde\balpha,\tilde\bbr)+10d_6(\tilde\balpha,\tilde\bbr)\\
	&\quad +(d_1(\tilde\balpha,\tilde\bbr)+d_2(\tilde\balpha,\tilde\bbr))+6(d_1(\tilde\balpha,\tilde\bbr)+d_2(\tilde\balpha,\tilde\bbr)+d_4(\tilde\balpha,\tilde\bbr))\\
	&\quad+8\sum_{i=1}^3d_i(\tilde\balpha,\tilde\bbr) +  3\sum_{i=3}^5d_i(\tilde\balpha,\tilde\bbr)+2\sum_{i=4}^6d_i(\tilde\balpha,\tilde\bbr)+4\sum_{i=3}^6d_i(\tilde\balpha,\tilde\bbr)+\sum_{i=1}^5d_i(\tilde\balpha,\tilde\bbr). \label{eq:c6_1}
\end{align}
Next, without loss of generality we assume
\begin{align}
	0 = \tilde{r}_1 \geq \tilde{r}_2 \geq \cdots \geq \tilde{r}_6, \label{eq:assumedelta_K=6}
\end{align}
for the same reason as we have for the assumption (\ref{eq:assumedelta_K=5}) in the $K=5$ case in Section \ref{section: K5}.
We then apply (\ref{eq:assumedelta_K=6}) and the bounds $B_{[5]}$, $B_{[6]}$, $B_{[1,2]}$, $B_{[1,2,4]}$, $B_{[1,2,3]}$, $B_{[3,4,5]}$, $B_{[4,5,6]}$, $B_{[3,4,5,6]}$ and $B_{[1,2,3,4,5]}$ to the respective summands in (\ref{eq:c6_1}), and get
\begin{align}
	16 \dsumo(\tilde\balpha) \leq 16 \dsumo(\tilde\balpha) -\tilde{r}_2 \leq41 \dsumb(\tilde\balpha), \label{eq:c6_2}
\end{align}
where the first inequality holds because $\tilde{r}_2 \leq 0$. 
Since inequality (\ref{eq:c6_2}) holds for all $\balpha \in \mathcal{A}_6^+$ regardless of the assumption (\ref{eq:assumedelta_K=6}), we have $\sup_{\balpha \in \mathcal{A}_6^+} \frac{\dsumo(\balpha)}{\dsumb(\balpha)} \leq 41/16$. By Lemma \ref{lemma:mubound} we have $ \mu_6 \leq \max \{ \mu_5, 41/16\} = 41/16$, which is the desired converse bound.

Next, for achievability, let us consider the topology, say $\balpha_6$, depicted in Figure \ref{fig:extremalnet}(d). The direct link are set with strengths $\alpha_{11} = \alpha_{22} = 8, \alpha_{33}=10, \alpha_{44}= 12,$ and $\alpha_{55} = \alpha_{66} = 16$. The cross links emitting from Transmitter 3 have $\alpha_{i3} = 5$ (red solid), those from Transmitter 4 have $\alpha_{i4} = 6$ (blue dashed), those from Transmitter 5 have $\alpha_{i5} = 8$ (green dotted), and those from Transmitter 6 have $\alpha_{i6} = 10$ (orange dot-dashed). The GDoF tuple $(8,8,5,6,8,6)$ is achieved with power allocation $\bbr_6=(0,0,-5,-6,-8,-10)$; thus $ \dsumo(\balpha_6)\geq 41$. By exhaustive consideration of all $2^6$ BPC cases, we verify that $\dsumb(\balpha_6) =16$ with binary power control. Therefore, we establish that $\mu_6\geq \frac{41}{16}$ and complete the proof for the case of $K=6$. $\hfill\square$

\subsection{Proof of Corollary \ref{cor:2-4usr} } \label{sec:proofCor}
The extremal sum-GDoF gain, $\mu_K$, is already a lower bound for $\eta_K$, i.e., $\eta_K\geq \mu_K$, since $\mu_K$ is obtained only in the asymptotic high SNR regime ($P \rightarrow \infty$), while $\eta_K$ allows all SNR regimes.

For ease of exposition in finding bounds for $\eta_K$ $(K=2,3,4)$, we define the following notations. For a topology $\tbalpha \in \mathcal{A}_K$ of a $K$ interference network, let $\tbbr$ be any power allocation that achieves sum rate $\Rsumo(\tbalpha)$.  For $\usr \subset [K]$, we define $\bbr(\usr) = \left(r_1(\usr), r_2(\usr), \cdots, r_K(\usr)\right)$ be the power allocation with $r_i(\usr) = 0$ if $i \in \usr$ and $r_i(\usr) = -\infty$ if otherwise. 

The upper bound for $\eta_2 \leq 1$ is straightforward because BPC maximizes the sum rate when interference is treated as noise in a two-user interference network, regardless of its topology \cite{2user_1,2user_2}. In the following we use this fact to find upper bounds for $\eta_3$ and $\eta_4$ respectively. 
For all $\tbalpha \in \mathcal{A}_3$ and all $P \geq 0$,
\begin{align}
	&2 \Rsumo(\tbalpha, P) \notag\\
	&=  (R_{1}(\tbalpha, P, \tbbr) + R_{2}(\tbalpha, P, \tbbr)) + (R_{2}(\tbalpha, P, \tbbr) + R_{3}(\tbalpha, P, \tbbr)) + (R_{1}(\tbalpha, P, \tbbr) + R_{3}(\tbalpha, P, \tbbr)) \nonumber \\
	&\leq \max_{\usr \subset \{1,2\}} R_\Sigma(\tbalpha, P, \bbr(\usr)) + \max_{\usr \subset \{2,3\}} R_\Sigma(\tbalpha, P, \bbr( \usr )) + \max_{\usr \subset \{1,3\}} R_\Sigma(\tbalpha, P, \bbr( \usr )) \label{eq:sumRate3}\\
	&\leq  3 \Rsumb(\tbalpha, P), \label{eq:sumRate3-1}
\end{align}
where (\ref{eq:sumRate3}) holds as BPC maximizes the sum rate among all power allocations in two-user interference networks. Inequality (\ref{eq:sumRate3-1}) leads to the bound $\eta_3 = \sup_{\balpha \in \mathcal{A}_3, P \geq 0} \frac{\Rsumo(\balpha, P)}{\Rsumb(\balpha, P)} \leq \frac{3}{2}$.
We apply the same trick to find an upper bound for $\eta_4$. 
For all $\tbalpha \in \mathcal{A}_4$ and all $P \geq 0$, we have
\begin{align}
	&\Rsumo(\tbalpha, P) \notag \\
	&= \left( R_{1}(\tbalpha, P, \tbbr) + R_{2}(\tbalpha, P, \tbbr) \right) + \left( R_{3}(\tbalpha, P, \tbbr) + R_{4}(\tbalpha, P, \tbbr) \right)  \nonumber \\
	&\leq  \max_{\usr \subset \{1,2\}} R_\Sigma(\tbalpha, P, \bbr(\usr)) + \max_{\usr \subset \{3,4\}} R_\Sigma(\tbalpha, P, \bbr(\usr)) \label{eq:sumRate4} \\
	&\leq  2 \Rsumb(\tbalpha, P), \nonumber
\end{align}
which leads to  $\eta_4 = \sup_{\balpha \in \mathcal{A}_4, P \geq 0} \frac{\Rsumo(\balpha, P)}{\Rsumb(\balpha, P)} \leq 2$. 
Here we conclude the proof. $\hfill\square$

\section{Proof of Theorem \ref{thm:OsqrtK} } \label{sec:proofSqrtK}
First we show $\mu_K/\sqrt{K} \leq 5/2$ for all $K \in \mathbb{N}$. This upper bound holds for $K = 2$ and $3$ by Theorem \ref{thm:2-6usr}. To argue it holds for $K \geq 4$ as well, we apply Lemma \ref{lemma:mubound} for arbitrary $K \geq 4$, where $ m-1 < \sqrt{K} \leq m$ for some integer $m \geq 2$, and obtain,
\begin{align}
	\frac{\mu_K}{\sqrt{K}} 
	&\leq \max \left \{ \frac{\mu_{K-1}}{\sqrt{K}}, \frac{1}{\sqrt{K}}\sup_{\balpha \in \mathcal{A}_K^+ }\frac{\dsumo(\balpha)}{\dsumb(\balpha)} \right \} \label{eq:muK-1}\\
	&\leq \max \left \{ \frac{\mu_{K-1}}{\sqrt{K-1}}, \frac{1}{m-1} \sup_{\balpha \in \mathcal{A}_{m^2}^+ } \frac{\dsumo(\balpha)}{\dsumb(\balpha)} \right \} \label{eq:muK-2}\\
	&\leq \max \left \{ \frac{\mu_{K-1}}{\sqrt{K-1}}, \frac{\frac{3}{2}m-\frac{1}{2}}{m-1} \right \} \label{eq:muK-3}\\
	&\leq \max \left \{ \frac{\mu_{K-1}}{\sqrt{K-1}}, \frac{5}{2} \right \}. \label{eq:muK-4}
\end{align}
First we apply Lemma \ref{lemma:increase} to obtain (\ref{eq:muK-2}).
Then we apply (\ref{eq:K^2ubd}) from Lemma \ref{lemma:K^2bounds} to obtain (\ref{eq:muK-3}). 
Inequality (\ref{eq:muK-4}) follows because $m \geq 2$.

Next, we can argue from (\ref{eq:muK-4}) and Theorem \ref{thm:2-6usr} that $\frac{\mu_K}{\sqrt{K}} \leq 5/2$ for all $K \in \mathbb{N}$. 
To argue this by contradiction, suppose for some integer $\ell > 4$ we have $\frac{\mu_{\ell}}{\sqrt{\ell}} > 5/2$.
Then by applying (\ref{eq:muK-4}) we have $\max \left \{ \frac{\mu_{\ell-1}}{\sqrt{\ell-1}}, 5/2 \right \} \geq \frac{\mu_{\ell}}{\sqrt{\ell}} > 5/2$. The maximum on the leftmost side cannot be $5/2$, so we have $\frac{\mu_{\ell-1}}{\sqrt{\ell-1}} > 5/2$. By induction, we can deduce $\frac{\mu_{4}}{\sqrt{4}} > 5/2$. But this contradicts  Theorem \ref{thm:2-6usr}, which already established that $\frac{\mu_{4}}{\sqrt{4}} = 1$. 
As a result, $\frac{\mu_{K}}{\sqrt{K}} \leq 5/2$ for all $K \in \mathbb{N}$. 

Finally, we establish the lower bound in Theorem \ref{thm:OsqrtK}.  For $K$ satisfying $m-1 < \sqrt{K} \leq m$, we have
\begin{align}
	\frac{\mu_K}{\sqrt{K}} &\geq \frac{m-1}{\sqrt{K}} \geq \frac{m-1}{m},
\end{align}
where the first inequality is due to (\ref{eq:K^2lbd}) in Lemma \ref{lemma:K^2bounds}.
As a result, we have $\liminf_{K \rightarrow \infty} \frac{\mu_K}{\sqrt{K}} \geq 1$. Here we conclude the proof. $\hfill\square$

\section{Conclusion} \label{sec:conclusion}
Using  ideas from GDoF analyses and extremal network theory, we studied the extremal gain of  optimal power control over binary power control  especially in large interference networks, in search of theoretical counterpoints to well established insights from numerical studies. Whereas numerical studies have established that in most practical settings binary power control is close to optimal \cite{Alouini_cell, Hong_D2D, Andrews_D2D, Schober_D2D, Lozano_D2D}, our extremal analysis shows not only that there exist settings where the gain from optimal power control can be quite significant, but also  bounds the extremal values of such gains from a GDoF perspective. As our main contribution, we explicitly characterize the extremal GDoF gain of optimal over binary power control as $\Theta\left(\sqrt{K}\right)$ for all $K$. For $K=2,3,4,5,6$ users, the precise extremal gain is found to be $1, 3/2, 2, 9/4$ and $41/16$, respectively. Networks shown to achieve the extremal gain may be interpreted as multi-tier heterogeneous networks. 

It is worthwhile to note that the findings of this work do not contradict conventional wisdom that binary power control is generally close to optimum. Indeed  numerical experiments suggest that such extremal gains are unlikely to be encountered in practice. Even in the heterogeneous multi-tier topologies that emerge as extremal networks, numerical experiments suggest that the extremal gains are manifested only at very high SNRs. Figure \ref{fig:BPC_sim} highlights this  sobering insight with a simple numerical simulation for the network topology depicted in Figure \ref{fig:extremalK} with $K=9$ users. This topology, labeled $\balpha_9$, achieves sum-GDoF gain $\sqrt{K} = 3$ according to the lower bound (\ref{eq:K^2lbd}) in Lemma \ref{lemma:K^2bounds}. In Figure \ref{fig:BPC_sim} we illustrate the sum rate achieved by two power control policies. The black dashed curve, labeled  $R_{\Sigma,b}$, is the sum rate achieved by binary power control; i.e. $R_{\Sigma,b} = \Rsumb(\balpha_9, P)$. On the other hand, the red solid curve, labeled with $R_{\Sigma}\eqref{eq:KKpwr}$, is the sum rate achieved by the power allocation vector specified in \eqref{eq:KKpwr}. The power allocation (\ref{eq:KKpwr}) is motivated by the high SNR setting and indeed it reaches $\dsumo(\balpha_9)$. Such a power allocation does not necessarily achieve $\Rsumo(\balpha_9, P)$ for finite $P$. This is reflected by the observation that $R_\Sigma$ is less than $\Rsumb$ at very low SNRs ($P \leq 5$dB). Nevertheless, here we still use $R_{\Sigma}\eqref{eq:KKpwr}$ as an approximation of $\Rsumo(\balpha_9, P)$ because they are close at high SNRs. The performance gap of the two sum rates  is quantified by the sum-rate gain $R_\Sigma\eqref{eq:KKpwr}/R_{\Sigma,b}$ (blue dot-dashed curve). The sum-rate gain grows as $P$ increases, and eventually approaches $3$ as $P$ goes to infinity, which is guaranteed by (\ref{eq:K^2lbd}). However, the sum-rate gain is relatively modest at moderate SNRs, e.g., it is below 2 when $P \leq 20$dB. Thus, while the extremal gain is valuable as a sharp theoretical limit, and the $\Theta(\sqrt{K})$ scaling is remarkable, it offers only a complementary perspective from asymptotic analysis, and not a contradiction to the conventional wisdom that binary power control is generally close to optimal in  practical settings.

\begin{figure}[t!]
	\begin{center}
		\includegraphics[width=0.75\textwidth]{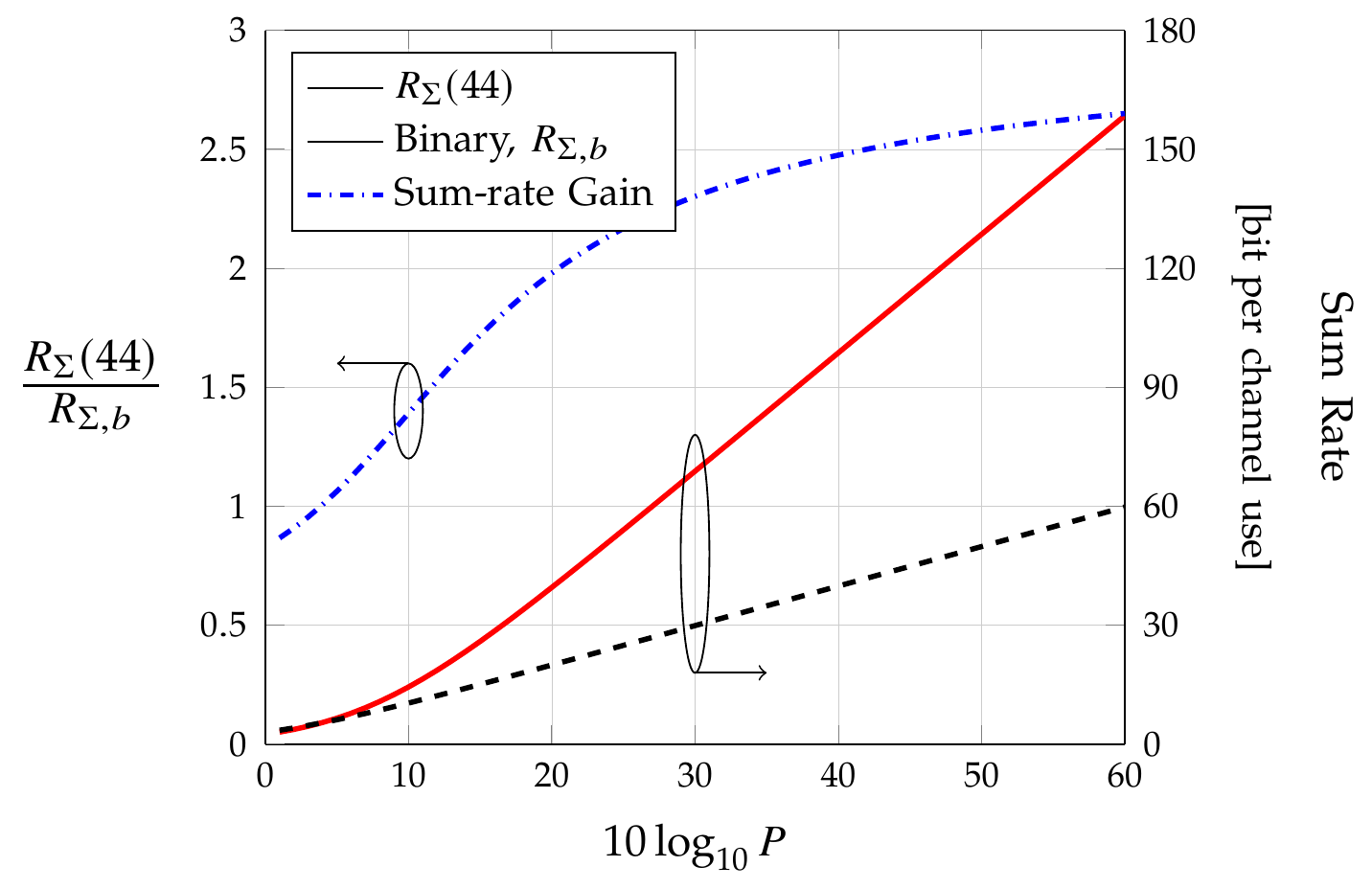}
		\caption{\small \it Sum rates achieved in the network depicted in Figure \ref{fig:extremalK} for $K = 9$ with two power control strategies: the GDoF-optimal power allocation scheme from (\ref{eq:KKpwr}) is used as a proxy for optimal power control and yields the red solid curve, denoted as $R_\Sigma\eqref{eq:KKpwr}$, while the optimal binary power control scheme yields the black dashed curve, denoted as $R_{\Sigma,b}$. Their performance gap is quantified with the sum-rate gain $\left. R_\Sigma\eqref{eq:KKpwr}\middle/ R_{\Sigma,b} \right.$ (blue dot-dashed curve). }
		\label{fig:BPC_sim}
	\end{center}
\end{figure}

\ifCLASSOPTIONcaptionsoff
  \newpage
\fi





\end{document}